\documentclass[12pt]{amsart}     

\usepackage{amsfonts}
\usepackage{amsmath}
\usepackage{graphicx}
\usepackage{amsfonts} \usepackage{amsmath, amsthm}
\usepackage{epsfig} \usepackage{verbatim}
\usepackage{amssymb, amsthm, amsmath, bm, tikz, enumerate, color}

\newtheorem{theorem}{Theorem}

\newtheorem{corollary}[theorem]{Corollary}

\newtheorem{definition}{Definition}

\newtheorem{lemma}{Lemma}

\newtheorem{proposition}{Proposition}

\begin{document}

\title[Local equilibrium in inhomogeneous models]{Local equilibrium in inhomogeneous stochastic models of heat transport} 
\author{P\'eter N\'andori}
\address{P\'eter N\'andori: Department of Mathematics, University of Maryland, College Park, MD 20742, USA}
\email{pnandori@math.umd.edu}

\begin{abstract} 
We extend the duality of Kipnis Marchioro and Presutti \cite{KMP82} to inhomogeneous 
lattice gas systems
where either the components have different degrees of freedom or the rate of interaction
depends on the spatial location.
Then the dual process is applied to prove local equilibrium 
in the hydrodynamic limit 
for some inhomogeneous high dimensional systems and
in the nonequilibrium steady state 
for one dimensional systems with arbitrary inhomogeneity.
\end{abstract}
\maketitle

\section{Introduction}

Describing the emergence of local equilibrium in systems forced out of equilibrium is one of the main
challenges of nonequilibrium statistical mechanics. Most mathematical works concern stochastic interacting 
particle systems as they are generally more tractable than the deterministic ones. 
Despite the big amount of recent
work in this field, little is known for systems with spatial inhomogeneity.

The present paper extends the classical duality result of Kipnis Marchioro and Presutti \cite{KMP82} 
to inhomogeneous systems, where the inhomogeneity means that either (A) the components of the system
have different degrees of freedom or (B) the interaction rate between the components depends on 
their location.
The dual process is roughly speaking a collection of 
biased random walkers that interact with one another once their distance is not bigger than $1$,
and are completely independent otherwise.
We leverage the dual process to prove the existence of local equilibrium in several cases. 
Some of these results were announced in \cite{LNY15}.
The systems we consider are somewhat similar to the ones recently introduced in \cite{CGRS15},
but are not the same.

\subsection{Informal description}
\label{sec:pr}

To fix notation, we denote by $\Gamma(\alpha, c)$ the probability distribution with density
$\frac{x^{\alpha -1} \exp(-x/c)}{c^{\alpha} \Gamma(\alpha)}$ for $x >0$, 
where $\alpha >0$ is the shape parameter and $c>0$ is the scale parameter ($\Gamma$ is the usual Gamma function).
The $k$th moment 
of the Gamma distribution is
\begin{equation}
\label{eq:gammamoment}
\int_0^{\infty} \frac{x^{k+\alpha -1} \exp(-x/c)}{c^{\alpha} \Gamma(\alpha)}
dx  = \frac{c^k \Gamma (\alpha + k)}{\Gamma (\alpha)}
\end{equation}
The $Beta(\alpha, \beta)$ distribution with parameters $\alpha , \beta >0$ has density 
$\frac{\Gamma(\alpha + \beta)}{\Gamma(\alpha) \Gamma(\beta)} x^{\alpha -1} (1-x)^{\beta -1}$
for $x>0$. For brevity, we will write $B(\alpha, \beta) = \frac{\Gamma(\alpha) \Gamma( \beta)}{\Gamma(\alpha + \beta)}$.\\

Our model can be informally described as follows. First consider two systems with degrees of freedom $\omega_i$
and
energy $\xi_i$ for $i=1,2$. 
Now let the two systems
exchange energies by a microcanonical procedure: redistribute the total energy according to the law of 
equipartition. Representing the 
energy per the $j$th degree of freedom by $X_j^2$ for $j=1,...,\omega_1 + \omega_2$, $\vec{X}$ is 
thus uniformly distributed
on the sphere $(\xi_1 + \xi_2) \mathbb S^{\omega_1 + \omega_2 -1}$.
Writing $\vec{X} = (\xi_1 + \xi_2) \vec{Y} / \| \vec{Y}\|$, where $\vec{Y}$ has $\omega_1 + \omega_2$ dimensional
standard normal distribution, the first system's energy is updated to
\[
\xi_1' = X_1^2 + ... +X_{\omega_1}^2 = \frac{(\xi_1 + \xi_2) (Y_1^2 + ... +Y_{\omega_1}^2)}
{Y_1^2 + ... +Y_{\omega_1 + \omega_2}^2}
\]
which is well known to have $(\xi_1 + \xi_2) Beta(\omega_1/2, \omega_2/2)$ distribution.

In the simplest case, our model consists of a one dimensional chain of systems located at sites $1, 2, ..., L-1$
and possibly having different degrees of freedom. Then we choose pairs of nearest neighbors randomly 
and update their energies
by the above rule. Furthermore, the systems at site $1$ and $L-1$ are coupled to heat baths of different
temperature. Then we study the macroscopic energy propagation and the emergence of local equilibrium.
The special case of constant $2$ degrees of freedom is \cite{KMP82}, while systems with (arbitrary) constant degrees
of freedom have been studied in \cite{CGRS15}.
A chain of alternating billiard particles ($2$ degrees of freedom) and pistons ($1$ degree of freedom), 
has been proposed in \cite{BGNSzT15}.
There, the energies are updated via a {\it much} more complicated deterministic rule.

Finally, we consider two more generalizations: (A) to higher dimensions and (B) to inhomogeneity in the rate 
of interaction.

\subsection{Random walks end electrical networks}

Here, we review the connection between one dimensional random walks and electrical networks as we will need
them in the study of the dual process. The connection extends to much more general graphs than $\mathbb Z^1$, see
e.g. \cite{K14}, Section 19.
Assume that
the weights (conductances) $w_{i+1/2}$, $i=A,A+1,...,B$, are given and a random walker $\mathcal S_k$
is defined on the set $\{ A, A+1,...,B\}$ by $\mathcal S(0) = I$ and  
$$\mathbb P( \mathcal S_k = i+1 | S_k = i) = \frac{w_{i+1/2}}{w_{i-1/2}+ w_{i+1/2}}, \quad
\mathbb P( \mathcal S_k = i-1 | S_k = i) = \frac{w_{i-1/2}}{w_{i-1/2}+ w_{i+1/2}}
$$
for all $i=A+1, ..., B-1$ (the states $A$ and $B$ are absorbing). 
Then the resistances $R_{i+1/2} = 1/w_{i+1/2}$ and the hitting probability
$$ p = \mathbb P (\min \{k: \mathcal S_k = 0\} < \min \{k: \mathcal S_k = N\} )$$
are connected by the well known formula
\begin{equation}
\label{eq:electric}
p = \frac{\sum_{i=I}^{B-1} R_i}{\sum_{i=A}^{B-1} R_i}.
\end{equation}

Finally, to fix terminology, 
we denote by SSRW the simple symmetric random walk, i.e. a random walk with independent
steps, uniformly distributed on the neighbors of the origin in $\mathbb Z^d$.

\subsection{Organization}
The rest of the paper is organized as follows. We define our model precisely in Section 
\ref{secdef}. In Section \ref{subsecdefdual}, the dual process is defined and the duality relation is proved.
In Section \ref{sec:hyd} the main results, namely local equilibrium in the hydrodynamic
limit in dimension $\geq 2$ and in the nonequilibrium steady state in dimension $1$, are formulated.
Section \ref{sec:5}
contains the proofs of 
the theorems concerning the hydrodynamic limit, namely
Theorems \ref{thm1} and \ref{thm2}, except for the proof of Proposition
\ref{lemmaN,y}. Then Proposition \ref{lemmaN,y} is proved in Section \ref{sec6} (except for the proof
of Lemma \ref{lemma:annuli}). Section \ref{sec:thm3} is the proof of Theorem \ref{thm3}, which is the case of 
nonequilibrium steady state. Finally, the Appendix contains the proof of Lemma \ref{lemma:annuli}.


\section{The model} \label{secdef}

Let us fix a dimension $d \geq 1$, and $\mathcal D \subset \mathbb R^d$, 
a bounded, connected open set for which 
$\partial \mathcal D$ is a piecewise $\mathcal C^2$ submanifold with no cusps. We prescribe a continuous function 
$T : \mathbb R^d \setminus \mathcal D \rightarrow \mathbb R_{+}$ to be thought of as temperature. 
For $L \gg 1$, the physical domain of our system is  
$$\mathcal D_{L} = L \mathcal D \cap \mathbb Z^d\ .$$
At each lattice point $z \in \mathcal D_L$ (which will be called a {\it site}) there is a physical system
of $\omega_z$ degrees of freedom. 
The rate of interaction
along the edge $e = (u,v) \in \mathcal E (\mathcal D_L)$ is denoted by $r_e = r_{(u,v)}$ (here, 
$\mathcal E (\mathcal D_L)$ is the set of edges of the lattice $\mathbb Z^d$ restricted
to $\mathcal D_L$). Finally, we also fix rates of interaction between boundary sites and
the heat bath: $r_{(u,v)}$ for $u \in \mathcal D_L, v \in \mathbb Z^d \setminus \mathcal D_L$
with $\| u-v\| =1$.

Throughout this paper, $|\cdot|$ denotes the cardinality of a finite set, 
and for $v \in \mathbb R^d$,  
$\langle v \rangle$ is its closest point in $\mathbb Z^d$.

\medskip
The time evolution of the energies $\bm X(t) = \bm X^{(L)}(t) = (\xi_{v}^{(L)}(t))_{v \in \mathcal D_{L}}$ is a
Markov process with generator
$$
(Gf)(\underline \xi) =  (G_1f)(\underline \xi) + (G_2f)(\underline \xi),
$$
where $G_1$ describes interactions within $\mathcal D_L$ and
$G_2$ stands for interaction with the bath. We define $G_1$ as follows.
There is an exponential clock at each edge $e  =(u,v) \in \mathcal E (\mathcal D_L)$
of rate $e_{(u,v)}$. When it rings,
the energies of the two corresponding systems ($\xi_u$ and $\xi_v$) are pooled together and redistributed 
according to a $Beta(\omega_u/2, \omega_v/2)$ distribution. Thus
\begin{equation*}
 (G_1f)(\underline \xi) =
\sum_{(u,v) \in \mathcal E (\mathcal D_L)} r_{(u,v)}
\int_0^1 \frac{1}{B \left( \frac{\omega_u}{2}, \frac{\omega_v}{2} \right)} p^{\frac{\omega_u}{2} -1}
(1-p)^{\frac{\omega_v}{2} -1} [f(\underline \xi') -f( \underline \xi)] dp 
\end{equation*}
where 
\begin{equation*}
\xi'_w = \left\{ \begin{array}{rl}
 \xi_w &\mbox{ if $w \notin \{ u,v\}$} \\
  p(\xi_{u} + \xi_v) &\mbox{ if $w = u$}\\
 (1-p)(\xi_{u} + \xi_v) &\mbox{ if $w = v$.}
       \end{array} \right.
\end{equation*}
Every edge $e = (u,v)$ where $u \in \mathcal D_L$ 
and $v \in \mathbb Z^d \setminus \mathcal D_L$
provides connection to the heat bath of temperature $T(v /L)$:
with rate $r_{(u,v)}$, $\xi_u$ is updated to $\Gamma \left( \frac{\omega_u}{2}, T \left( \frac{v}{L} \right)  \right)$.
That is,

\begin{equation*}
 (G_2f)(\underline \xi) =
\sum_{u \in \mathcal D_L, v \in \mathbb Z^d \setminus \mathcal D_L, \| u-v\| =1} r_{(u,v)}
\int_0^{\infty} 
\frac{\eta^{\frac{\omega_u}{2} -1} \exp \left[-\frac{\eta}{T\left( \frac{v}{L} \right)} \right]}{ 
\left[ T\left( \frac{v}{L} \right) \right]^{\frac{\omega_u}{2}}
\Gamma \left( \frac{\omega_u}{2} \right)} [f(\underline \xi'') -f( \underline \xi)] d \eta 
\end{equation*}
where 
\begin{equation*}
\label{eq:xihat}
\xi''_w = \left\{ \begin{array}{rl}
 \xi_w &\mbox{ if $w \neq u$} \\
  \eta &\mbox{ if $w = u$.}
       \end{array} \right.
\end{equation*}

This completes the definition of $\bm X^{(L)}(t)$.


\section{The dual process} \label{subsecdefdual}

As mentioned in the Introduction, we want to understand the asymptotic behavior of $\bm X(t)$ by 
switching to a dual process. This section is devoted to the discussion of duality.

For $\mathcal D \subset \mathbb R^d$ and $L$ as in Section \ref{secdef}, we now introduce
a Markov process ${\bm Y^{(L)}(t)} = ((n_v^{(L)} (t))_{v \in \mathcal D_L}, (\hat n_v^{(L)} (t))_{v \in \mathcal B_L})$ 
designed to carry certain dual object, which we call particles, 
from sites in $\mathcal D_L$ to 
$$\mathcal B_L = \{ v \in \mathbb Z^d \setminus \mathcal D_L: \exists v \in \mathcal D_L: \| u-v\| =1\}.$$
Here, $n_{v}$ the number of particles at site $v \in \mathcal D_L$ and
$\hat n_{w}$ the number of particles permanently drooped off to the storage at $w \in \mathcal B_L$.
The generator of the process ${\bm Y(t)}$ is given by 
$$ (Af)(\underline n) = (A_1f)(\underline n) + (A_2f)(\underline n),$$
where $A_1$ corresponds to movements inside $\mathcal D_L$ and $A_2$ corresponds to 
the process of dropping off the particles to the storage. That is,
\begin{eqnarray*}
&& (A_1f)(\underline n) =  
\sum_{(u,v) \in \mathcal E ( \mathcal D_L)} r_{(u,v)} \\
&& \sum_{k=0}^{n_v+n_w} {{n_u+n_v} \choose {k}} \frac{B(k+ \frac{\omega_u}{2}, n_u+n_v-k+\frac{\omega_v}{2})}{ B(\frac{\omega_u}{2},\frac{\omega_v}{2}) }
 [f(\underline n') - f(\underline n) ]
\end{eqnarray*}
where 
\begin{equation}
\label{eq:n'}
n'_w = \left\{ \begin{array}{rl}
 n_w &\mbox{ if $w \notin \{ u,v\}$} \\
  k &\mbox{ if $w = u$.}\\
  n_v+n_w - k &\mbox{ if $w = v$.}
       \end{array} \right.
\quad \text{ and  }\quad \hat n'_w = \hat n_w \forall w \in \mathcal B_L .
\end{equation}
Recall that in case of the process $\bm X$, the energies are redistributed according to
a beta distribution with parameters $\omega_u/2, \omega_v/2$. In case of the dual process $\bm Y$, 
we redistribute the particles with the so called beta binomial distribution: first we choose a $p$
according to $Beta(\omega_u/2, \omega_v/2)$, then we choose $n'_u$ with binomial distribution of 
parameters $n_u + n_v, p$ and
$n'_v= n_u+n_v-n'_u$.\\
The second part of the generator is given by 
$$
(A_2f)(\underline n) = \sum_{u \in \mathcal D_L, v \in \mathcal B_L,
\| u-v\| =1} r_{(u,v)} [f(\underline n') - f(\underline n) ]$$
where 
\begin{equation*}
n'_w = \left\{ \begin{array}{rl}
 n_w &\mbox{ if $w \neq u$} \\
  0 &\mbox{ if $w = u$.}
       \end{array} \right.
\quad \text{ and  }\quad 
\hat n'_w = \left\{ \begin{array}{rl}
 \hat n_w &\mbox{ if $w \neq v$} \\
  \hat n_v+ n_u &\mbox{ if $w = v$}
       \end{array} \right. 
\end{equation*}
This completes the definition of $\bm Y^{(L)} (t)$. \\

Now we turn to the duality. Let us define the function with respect to which the duality holds
\begin{equation}
 \label{eq:F}
F(\underline n, \underline \xi) = 
\prod_{u \in \mathcal D_L} \frac{\xi_u^{n_u}\Gamma(\omega_u/2)}{\Gamma(n_u+\omega_u/2)}
\prod_{v \in \mathcal B_L} \left[ T\left( \frac{v}{L} \right) \right]^{\hat n_v}
\end{equation}
The duality with respect to $F$ means that
\begin{proposition}
\label{prop:dual}
For any $\underline \xi$ any $\underline n$ and any $t>0$,
 $$ \mathbb E (F(\underline n, \bm X(t)) | \bm X_0 = \underline \xi) = 
\mathbb E (F( \bm Y(t), \underline \xi) | \bm Y_0 = \underline \xi) $$
\end{proposition}

\begin{proof} Clearly it is enough to prove that for any $\underline \xi$ and $\underline n$,
$$ G F (\underline \xi,\underline n) = A  F (\underline \xi,\underline n).$$
To prove this, we consider the following two cases, where Case $i$ corresponds to $G_i$ and $A_i$
for $i=1,2$.

\medskip
\noindent {\bf Case 1.} The clock on the edge $(u,v)$ rings. 
The term corresponding to $u$ and $v$ in $G_1F$ 
can be written as $r_{(u,v)} \cdot I \cdot II$, where
$$ I =
\Gamma(\omega_u/2)\Gamma(\omega_v/2)
\prod_{w \in \mathcal D_L \setminus \{ u,v\} } \frac{\xi_w^{n_w} \Gamma(\omega_w/2)}{\Gamma(n_w+\omega_w/2)}
\prod_{w \in \mathcal B_L} \left[ T\left( \frac{w}{L} \right) \right]^{\hat n_w}
 $$
and
$$
 II =
\frac{ \int_0^1  p^{n_u + \frac{\omega_u}{2} -1 } (1-p)^{n_v + \frac{\omega_v}{2} -1 } \frac{1}{B(\frac{\omega_u}{2}, \frac{\omega_v}{2})}
(\xi_u + \xi_v)^{n_u+n_v} dp
 - \xi_{u}^{n_{u}} \xi_{v}^{n_{v}} }{{\Gamma(n_u + \frac{\omega_u}{2})} \Gamma(n_v + \frac{\omega_v}{2})}
$$
Then we compute
\begin{eqnarray*}
 && II +  \frac{\xi_{u}^{n_{u}} \xi_{v}^{n_{v}} }{{\Gamma(n_u + \frac{\omega_u}{2})} \Gamma(n_v + \frac{\omega_v}{2})} =\\
&=& \frac{1}{B(\frac{\omega_u}{2}, \frac{\omega_v}{2})} \frac{1}{\Gamma (n_u + n_v +  \frac{\omega_u + \omega_v}{2})} (\xi_u + \xi_v)^{n_u+n_v}\\
&=& \frac{1}{B(\frac{\omega_u}{2}, \frac{\omega_v}{2})} \frac{1}{\Gamma (n_u + n_v +  \frac{\omega_u + \omega_v}{2})} 
\sum_{k=0}^{n_u+n_v} {{n_u+n_v} \choose k}\xi_u^{n_u} \xi_v^{n_v}\\
&=& \frac{1}{B(\frac{\omega_u}{2}, \frac{\omega_v}{2})}  
\sum_{k=0}^{n_u+n_v} {{n_u+n_v} \choose k}
B \left( k+ \frac{\omega_u}{2}, n_u+n_v-k+\frac{\omega_v}{2} \right)
\frac{\xi_u^{n_u}}{\Gamma(k+ \frac{\omega_u}{2})} \frac{\xi_v^{n_v}}{\Gamma (n-k+ \frac{\omega_v}{2})}
\end{eqnarray*}

Thus $r_{(u,v)} \cdot I \cdot II$ is the term corresponding to $u$ and $v$ in $A_1F$.

\medskip
\noindent {\bf Case 2.} The energy at site $u \in \mathcal D_L$ is updated by the heat bath
at $v \in \mathcal B_L$ (where $\| u-v\|=1$). 
As before, we write the term corresponding to $(u,v)$ in $G_2F$ as $r_{(u,v)} \cdot I \cdot II$, where 
$$ I =\prod_{u' \in \mathcal D_L \setminus \{ u\} } \frac{\xi_{u'}^{n_{u'}} 
\Gamma(\omega_{u'}/2)}{\Gamma(n_{u'}+\omega_{u'}/2)}
\prod_{v' \in \mathcal B_L \setminus \{ v\}} \left[ T\left( \frac{v'}{L} \right) \right]^{\hat n_{v'}}
 $$
and
$$
 II =
\frac{\Gamma(\frac{\omega_u}{2})}{\Gamma(n_u + \frac{\omega_u}{2})}
\left[
\int_0^{\infty} 
\frac{\eta^{n_u+\frac{\omega_u}{2} -1} \exp \left[-\frac{\eta}{T\left( \frac{v}{L} \right)} \right]}{ \left[ T\left( \frac{v}{L} \right) \right]^{\frac{\omega_u}{2}}
\Gamma \left( \frac{\omega_u}{2} \right)} d \eta 
 - \xi_{u}^{n_{u}} \right] \left[ T\left( \frac{v}{L} \right) \right]^{\hat n_v}
$$
By (\ref{eq:gammamoment}), we obtain that 
\begin{eqnarray*}
&& II + \frac{\xi_u^{n_u} \Gamma(\frac{\omega_u}{2})}{\Gamma(n_u + \frac{\omega_u}{2})}\left[ T\left( \frac{v}{L} \right) \right]^{\hat n_v}
= \left[ T\left( \frac{v}{L} \right) \right]^{\hat n_v+n_u} \\
\end{eqnarray*}
Thus $r_{(u,v)} \cdot I \cdot II$ is the term corresponding to $v$ in $A_2F$.
\end{proof}

Note that the process $\bm Y$ preserves the total number of particles, which will be denoted be $N$.
We conclude this section with the following simple lemma.

\begin{lemma}
\label{lemma:subsyst}
The restriction of the process $\bm Y$ to arbitrary subset of $K$ particles (with $K<N$) is also a Markov
process and satisfies the definition of $\bm Y$ with $N$ replaced by $K$.
\end{lemma}

\begin{proof} Without loss of generality, we can consider a system with
$K=N-1$ particles, assume that the clock attached to the edge $(v,w)$ rings and 
 the union of the particles at sites $u$ and $v$ prior to the mixing are labeled $\{1, \dots, n\}$. 
 
 Then
the probability that after the mixing, the set of particles at site $u$ is 
exactly $\{j_1, ...j_l\} \subset \{1,2,...,n\}$ is given by
\begin{equation}
\label{eq:subsys1}
 p(j_1, ..., j_l) = {n \choose l} 
\frac{\Gamma(l+\omega_u/2)\Gamma(n-l+\omega_v/2)}{\Gamma( n + \omega_u/2 + \omega_v/2)}  
\frac{\Gamma(\omega_u/2 +\omega_v/2)}{\Gamma( \omega_u/2)\Gamma( \omega_v/2)} 
\frac{1}{{n \choose l}}
\end{equation}
Now assume we add a new particle (of index $N$). If this new particle is not at sites $u$
or $v$, then clearly the situation is not disturbed.
If it is there, we compute
\begin{eqnarray}
&& p(j_1, ..., j_l) + p(j_1, ..., j_l,N+1) = \label{eq:subsys2} \\
&&\frac{\Gamma(l+\omega_u/2)\Gamma(n+1-l+\omega_v/2)}{\Gamma( n + 1+ \omega_u/2 + \omega_v/2)}  
\frac{\Gamma(\omega_u/2 +\omega_v/2)}{\Gamma( \omega_u/2)\Gamma( \omega_v/2)} \nonumber \\
&&+
\frac{\Gamma(l+1+\omega_u/2)\Gamma(n-l+\omega_v/2)}{\Gamma( n+1+ \omega_u/2 + \omega_v/2)}  
\frac{\Gamma(\omega_u/2 +\omega_v/2)}{\Gamma( \omega_u/2)\Gamma( \omega_v/2)} \nonumber
\end{eqnarray}
An elementary computation shows that (\ref{eq:subsys1}) is equal to (\ref{eq:subsys2}).  
The lemma follows.
\end{proof}


\section{Local Equilibrium}
\label{sec:hyd}

Let $d$, $\mathcal D$ and $L$ be as before such that $\mathcal D_L$ is connected.

First we state the existence and uniqueness of invariant measure in the equilibrium case.

\begin{proposition}
\label{prop:equilibrium}
Fix arbitrary functions $\omega: \mathcal D_L \rightarrow \mathbb Z_+$ 
and $r: \mathcal E (\mathcal D_L) \rightarrow \mathbb R_+$. 
If $T$ is constant, then 
$$ \mu^{(L)}_e = \prod_{v \in \mathcal D_L} \Gamma \left( \frac{\omega_v}{2}, \frac{1}{T} \right)$$
is the unique invariant probability measure of the process $\bm X^{(L)}(t)$.
\end{proposition}

\begin{proof}
The discussion in Section \ref{sec:pr} implies the following statement (which is actually well known, see e.g.
Lemma 3 in \cite{CHS81}).
 Let $\xi_1$ and $\xi_2$ be 
 independent Gamma distributed random variables with shape parameter $k$ and $l$, respectively 
and with the same scale parameter. Let $Z$ be independent from $X$ and $Y$ and have Beta 
distribution with parameters $k/2$ and $l/2$. 
Then the pair $(Z(\xi_1+\xi_2), (1-Z)(\xi_1+\xi_2))$ has the same distribution as $(\xi_1,\xi_2)$.
Proposition \ref{prop:equilibrium} follows.
\end{proof}

Our primary interest is in the out-of-equilibrium settings where the
bath temperature is non constant:

\begin{proposition}
\label{prop:invmeasexists}
Let 
$\omega: \mathcal D_L \rightarrow \mathbb Z_+$ 
and $r: \mathcal E (\mathcal D_L) \rightarrow \mathbb R_+$
be arbitrary. The process ${\bm X^{(L)}(t)}$ has a unique invariant probability measure $\mu^{(L)}$.
Furthermore, the distribution of ${\bm X^{(L)}(t)}$ converges to $\mu^{(L)}$ as $t \rightarrow \infty$
for any initial distribution of ${\bm X^{(L)}(0)}$.
\end{proposition}

We skip the proof of Proposition
\ref{prop:invmeasexists} since it is very similar to 
the analogous propositions in earlier similar models, see 
Proposition 1.2 in \cite{RY07}, Proposition 2 in
\cite{LNY15}.

\subsection{Hydrodynamic limit}

In order to discuss the local equilibrium in the hydrodynamic limit, 
we need some definitions. 
First we introduce some properties of the initial measures.

\begin{definition}
We say that $\bm X (0)$ is associated with $f$ if
for any fixed $\delta$ and any $k$
\begin{equation}
\label{eq:init}
\mathbb E \left( \prod_{i=1}^k \xi_{v_i}^{(L)} (0) \right) \sim \prod_{i=1}^k \frac{\omega_{v_i}}{2} f(v_i/L) 
\end{equation}
as $L \rightarrow \infty$ 
uniformly for every $v_1, ..., v_k \in \mathcal D_L$ satisfying $\| v_i - v_j\| \geq \delta L $ for $i \neq j$
and with some fixed continuous function $f : \mathbb R^d \rightarrow \mathbb R_+$ 
such that $f|_{\mathbb R^d \setminus D} = T$
\end{definition}



\begin{definition}
We say that $\bm X (0)$ satisfies the uniform moment condition if there
are constants $C_k$ such that $\mathbb E ( \xi_{v}^k (0)) < C_k$ for every $L$ and for every 
$v \in \mathcal D_L$.
\end{definition}

Recall that the L\'evy-Prokhorov distance is the metrization of weak convergence of measures.

\begin{definition}
We say that $\bm X^{(L)}(t)$ approaches local equilibrium 
in the hydrodynamic
limit
at $x \in \mathcal D$ and $t>0$  if for any finite set 
$S \subset \mathbb Z^d$ 
the L\'evy-Prokhorov distance of the distribution of $\bm X^{(L)}(tL^2)$ 
restricted to the components $(\xi_{\langle xL \rangle+s})_{s \in S}$
and 
 $$\prod_{s \in S} \Gamma \left( \omega_{{\langle xL \rangle+s} }/2, u(t,x) \right)$$ 
converges to zero as $L \rightarrow \infty$.
\end{definition}

We will choose the initial distributions, i.e. the distributions of $\bm X^{(L)}(0)$, which are
associated with a continuous function $f$. 
The interesting question is that 
what kind of equation defines $u$ for different choices of $\omega$ and $r$. In any case, we 
expect that the initial condition is given by $f$ and the boundary condition by $T$. 
We will consider the two simplest cases here.

\begin{theorem}
\label{thm1}
Assume $d \geq 2$ and $\omega_v=\omega_0 \in \mathbb Z_+$ for every site $v$.
Assume furthermore that 
$r_{(u,v)}=R(\frac{u+v}{2L})$ for every $u \in \mathcal D_L$ and $v \in \mathbb Z^d$ with
$\| u-v\|=1$, where $R \in \mathcal C^2 (\mathbb R^d, \mathbb R_+)$.
Also assume that $\bm X^{(L)}(0)$ is associated with $f$,
a continuous function $f : \mathcal {\bar D} \rightarrow \mathbb R_+$,
and satisfy the uniform moment condition. Then
$\bm X^{(L)}(t)$ approaches local equilibrium in the hydrodynamic
limit for all $x \in \mathcal D$ and $t>0$ with $u$ the unique solution of the equation
\[
\begin{cases}
 u_t =  \nabla( R\nabla u), \\
 u(0,x) = f(x), \\
 u(t,x)|_{\partial \mathcal D} = T(x).
\end{cases}
\]
\end{theorem}
 
{\bf Remark about the initial conditions}
 Note that $f$ represents the energy per degrees of freedom at time zero, that is
 why we need the multiplier $\omega_v/2$ on the right hand side of (\ref{eq:init}). 
 We do not have to assume
local equilibrium at zero, which would correspond to the special choice of $\bm X(0)$: the product
of Gamma distributions (with shape parameter $\omega_v/2$, and scale parameter $f(v/L)$). 
One interesting
consequence of Theorem \ref{thm1} (and similarly that of Theorem \ref{thm2}) is that
the system satisfies the local equilibrium for arbitrary positive macroscopic time even if 
it only satisfies the given weaker condition at time zero.
Since we want to leverage the duality via moments, we also need to assume some condition
on the higher moments. The simplest one is the uniform moment condition. Most probably, neither 
condition (\ref{eq:init}) nor the uniform moment condition is optimal,
but we do not pursue the most general case here. 
 
Our next choice is the simplest non-continuous environment: 
we consider $\mathcal D = [-1,1]^d$ with one of the functions 
$\omega$ and $r$ being constant on $\mathcal D$ and the other one is constant on 
 $[-1,0] \times [-1,1]^{d-1}$ and $[0,1] \times [-1,1]^{d-1}$. Thus we have the following

\begin{theorem}
\label{thm2}
 Let $d \geq 2$, $\mathcal D = [-1,1]^d$.
Assume that $\bm X^{(L)}(0)$ is associated with $f$,
a continuous extension of $T$ to $\mathcal {\bar D}$,
and satisfies the uniform moment condition
\begin{enumerate}
\item[(a)] Let $\omega_v = \omega_{-1}$ if $v_1 <0$ 
and $\omega_v = \omega_{1}$ if $v_1 \geq 0$ with some positive
integers $\omega_{-1}, \omega_1$ and let $r$ be constant.
 Then
$\bm X^{(L)}(t)$ approaches local equilibrium in the hydrodynamic
limit for all $x \in \mathcal D$ with $x_1 \neq 0$ and all $t>0$ with 
$u$ the unique solution of the equation
\[
\begin{cases}
u_t = r \Delta u \quad \text{ for $x \in \mathcal D$ with $x_1\neq 0$},\\
\omega_1 \frac{\partial}{\partial x_1+} u(t,(0,x_2, ..., x_d)) = 
\omega_{-1}  \frac{\partial}{\partial x_1 -} u(t,(0, x_2, ..., x_d))\\
u(0,x) = f(x), \\
u(t,x)|_{\partial \mathcal D} = T(x). 
\end{cases}
\]
\item[(b)]
Let $r_{(u,v)} = r_{-1}$ if $u_1+v_1 <0$ and 
$r_{(u,v)} = r_{1}$ otherwise, where $r_{-1}, r_1$ are fixed positive numbers. 
Similarly, $r_w = r_{-1}$ if $w_1 <0$ and $r_w = r_1$ otherwise. Let $\omega$ be constant.
 Then
$\bm X^{(L)}(t)$ approaches local equilibrium in the hydrodynamic
limit for all $x \in \mathcal D$ with $x_1 \neq 0$ and all $t>0$ 
with $u$ the unique solution of the equation
\[
\begin{cases}
u_t = r_{sign(x_1)} \Delta u \quad \text{ for $x \in \mathcal D$ with $x_1\neq 0$},\\
r_1 \frac{\partial}{\partial x_1+} u(t,(0,x_2, ..., x_d)) = 
r_{-1}  \frac{\partial}{\partial x_1 -} u(t,(0, x_2, ..., x_d))\\
u(0,x) = f(x), \\
u(t,x)|_{\partial \mathcal D} = T(x). 
\end{cases}
\]
\end{enumerate}

\end{theorem}

\subsection{Nonequilibrium steady state}

Now we are interested in 
the invariant measure of the Markov chains for finite (but large) $L$. 
Specifically, we are looking for a function $u(x), x \in \mathcal D$ such that in the limit
$\lim_{L \rightarrow \infty} \lim_{t \rightarrow \infty}$
the local temperature exists and is given by $u(x)$. In case the hydrodynamic limit is known, 
$u(x)$ is expected to be equal to $\lim_{t \rightarrow \infty} u(t,x)$. We will choose arbitrary 
environment ($r$ and $\omega$), that is why we define the local equilibrium in a little more 
general form, namely with an $L$ dependent $u$. Of course in all natural examples,
one expects $u^{(L)}$ to converge.

\begin{definition}
We say that $\bm X^{(L)}(t)$ approaches local equilibrium in the nonequilibrium steady state
if for any $x \in \mathcal D$ and any finite set 
$S \subset \mathbb Z^d$
the L\'evy-Prokhorov distance of the invariant measure of $\bm X^{(L)}(t)$ 
restricted to the components $(\xi_{\langle xL \rangle+s})_{s \in S}$
and 
 $$\prod_{s \in S} \Gamma \left( \omega_{{\langle xL \rangle+s} }/2, u^{(L)}(x) \right)$$ 
converges to zero as $L \rightarrow \infty$.
\end{definition}

Let us fix $d=1$ and $\mathcal D = (0,1)$. To simplify notation, we will write $r_{m+1/2} := r_{(m,m+1)}$,
$\omega_0 := \omega_1$ and $\omega_L := \omega_{L-1}$.
Furthermore, we will need the following two definitions:
$$
\psi(m) = \frac{\omega_{m-1} + \omega_m}{r_{m-1/2} \omega_{m-1} \omega_m} \quad
\text{for } 1 \leq m \leq L 
$$
and
$$
\mathcal A^{(L)}(x) = \frac{\sum_{m=1}^{\lfloor xL \rfloor} 
\psi(m)}{\sum_{m=1}^{L} 
\psi(m)}.
$$

\begin{theorem}
\label{thm3}
Let $d=1$ and $\mathcal D = (0,1)$. Assume that the functions $r$ and $\omega$ are  
bounded away from zero and infinity uniformly in $L$ and the temperature on the boundary is given by 
$T(0), T(1) \in \mathbb R_+ \cup \{ 0\}$. Then 
$\bm X^{(L)}(t)$ approaches local equilibrium in the nonequilibrium steady state with 
$$ u^{(L)}(x) = 
(1 - \mathcal A^{(L)} (x)) T(0) + \mathcal A^{(L)}(x) T(1).
$$
\end{theorem}

Clearly, $u^{(L)}(x)$ can easily diverge in this generality. That is why we consider two special cases.
In case (a), $r$ is constant and $\omega$  is random. In
this case, we prove the quenched local equilibrium in the nonequilibrium steady state,
i.e. the almost sure convergence of $ u^{(L)}(x)$ to a deterministic limit. 
In case (b), $\omega$ is constant but $r$ is prescribed by a non-constant macroscopic function. 
We also prove that $ u^{(L)}(x)$ converges in this case.

\begin{proposition}
\label{prop:randomomega}
\begin{enumerate}
\item[(a)]
Let $r$ be constant, $K$ be the maximal degrees of freedom and let us
fix some continuous function 
$$ \varkappa : [0,1] \rightarrow \{p \in \mathbb R^K: p \geq 0, \sum_{i=1}^K p_i = 1 \}.$$
For each $L$, let us choose $\omega_{v}^{(L)}$ randomly and independently from
one another with 
$$\mathbb P(\omega_{v}^{(L)} = i) = \varkappa_i(v/L).$$
Then for almost every realization of the random functions $\omega^{(L)}$,
$$\lim_{L \rightarrow \infty} \mathcal A^{(L)}(x) = 
\frac{\sum_{i=1}^K \frac{1}{i} \int_0^x \varkappa_i(y)dy}{\sum_{i=1}^K \frac{1}{i} \int_0^1 \varkappa_i(y)dy}.$$
\item[(b)]
Let $\omega$ be constant, $\varrho: [0,1] \rightarrow \mathbb R_+$ is a continuous
function. 
For each $L$ and $v = 0,1,...L-1$ define
$r^{(L)}_{v+1/2} = \varrho(v/L)$.
Then 
$$\lim_{L \rightarrow \infty} \mathcal A^{(L)}(x) = 
\frac{\int_0^x 1/ \varrho(y)dy}{ \int_0^1 1/ \varrho(y)dy}.$$
\end{enumerate}

\end{proposition}

Before turning to the proofs of the above results, we briefly comment on some possibilities of extension.


\subsection{Possible extensions}

As we will see, the proof of Theorems \ref{thm1} and \ref{thm2} also provides the local equilibrium
in the nonequilibrium steady state.

\begin{corollary}
Consider the setup of either Theorem \ref{thm1} or \ref{thm2}.  Then $\bm X^{(L)} (t)$ approaches local 
equilibrium (assuming $x_1 \neq 0$ in case of Theorem \ref{thm2})
in the nonequilibrium steady state with $U(x) = \lim_{t \rightarrow \infty} u(t,x)$.
\end{corollary} 

Indeed, 
for $\delta >0$ fixed one can find some $t$ large such that with probability at least 
$1-\delta$, all processes in Proposition \ref{lemmaN,y} 
have arrived at the 
boundary before $tL^2$. In such cases, $\tilde{\bm Y}_i^{(L)}({tL^2}) = \tilde{\bm Y}_i^{(L)}({\infty})$, 
and the latter can be used to 
 prove local equilibrium in the nonequilibrium steady state (cf. the proof of Theorem \ref{thm3}).
 
Furthermore, it seems likely that 
our proof could be adapted to a version of
Theorem \ref{thm2} with more general domains and with piecewise constant
$r$ and $\omega$ (see \cite{P90} for the extension of skew Brownian motion
to such scenarios). 

However, the case of more general inhomogeneity in either high dimensions or in the hydrodynamic limit 
can be difficult. For example, if $\varkappa$ is constant in Proposition \ref{prop:randomomega}(a), then 
in order to verify the hydrodynamic limit, one would have to compute the scaling limit of some interacting 
random walkers among iid conductances, whereas even the case of one random walker is non obvious 
(see \cite{SSz04}). Clearly, the case of high dimensions or non iid environments are even much harder.


\section{Proof of Theorems \ref{thm1} and \ref{thm2}}
\label{sec:5}

Since the proofs of Theorems \ref{thm1} and \ref{thm2} are very similar, we provide one
proof and distinguish between cases Theorems \ref{thm1}, \ref{thm2}(a) and \ref{thm2}(b)
if necessary.
Let us fix some $t>0$, a point $x \in \mathcal D$ and a finite set $S \subset \mathbb Z^d$. 
We need to show that the (joint) distribution of $(\xi_{\langle xL \rangle +s} (tL^2))_{s \in S}$
converges to the product of gamma distributions with the scale parameter $u(t,x)$.
As it is well known, any product of Gamma distributions is characterized by its moments (see \cite{Z83}). 
Thus it is enough
to prove that the moments of $(\xi_{\langle xL \rangle +s})_{s \in S}$ converge to the product of
the moments of gamma distributions 
(convergence of second moments implies tightness). When computing the 
moments of order $n^*_s \in \mathbb N, s\in S$, we can use Proposition \ref{prop:dual} to switch to the dual 
process. 

We need to introduce some auxiliary processes. 
Let us denote by $\tilde{\bm Y}$ the slight variant of $\bm Y$ where the position
of distinguishable particles are recorded. More precisely, the 
phase space of $\tilde{\bm Y} = (\tilde{\bm Y}^{(L)}_{i}(t))_{0 \leq t, 1 \leq i \leq N}$ 
is $(\mathcal D_L \cup \mathcal B_L)^N$,
where $N=\sum_{s\in S} n^*_s$
and the initial condition is given such
that for all $s \in S$,
$$ \# \{ i: \tilde{\bm Y}_{i}(0) =  \langle xL \rangle + s \} = n^*_s.$$
For any $t\geq 0$ we define $\tilde{\bm Y} (t)$ so that
\begin{equation}
\label{deftildeY}
 \# \{ i: \tilde{\bm Y}_{i}(t) =  v \in \mathcal D_L\} = n_v(t_k),  \quad
\# \{ i: \tilde{\bm Y}_{i}(t) =  v \in \mathcal B_L \} = \hat n_v(t_k)
\end{equation}
and for all $t$,
$$  \# \{ i: \tilde{\bm Y}_{i}(t-) \neq  \tilde{\bm Y}_{i}(t+) \} \leq 1.$$

We want to show that the diffusively rescaled version of $\tilde{\bm Y}$ converges weakly to 
$N$ independent copies of some (generalized) diffusion processes $\mathcal Y$. Then the Kolmogorov 
backward equation associated with these processes will provide the function $u$. 
Clearly, the process $\mathcal Y$ will have to be stopped on $\partial \mathcal D$. 
So as to shed light on the main component of the proof, namely the convergence to the diffusion
process, we introduce a further simplification by not stopping the particles.

We can assume that $R$ is bounded in case of Theorem \ref{thm1}
(possibly by multiplying with a smooth function which is constant on $\mathcal D$ and decays quickly) 
and extend the definition of $\omega_v$, $r_{(u,v)}$
in case of Theorem \ref{thm2} for any $u,v \in \mathbb Z^d$. Then we consider the process
$\bm Z (t) = (n_v(t))_{v \in \mathbb Z^d}$ on the space 
\[z_v \in \mathbb Z_+, \quad \sum_{v \in \mathbb Z^d} z_v = N
\]
with the generator $A'_1$, obtained from $A_1$ by replacing 
$\sum_{(u,v) \in \mathcal E(\mathcal D_L)}$
with $\sum_{(u,v) \in \mathcal E(\mathbb Z^d)}$. Then we define $\tilde{\bm Z}$ from $\bm Z$
the same way as we defined $\tilde{\bm Y}$ from $\bm Y$. Observe that by construction
(and with the natural coupling) for $N=1$ we have
\begin{equation}
\label{eq:YZ1}
 \tilde{\bm Y}^{(L)} (s) = \tilde{\bm Z} (s \wedge \tau_{\mathcal D_L}) \quad \text{where} \quad 
\tau_{\mathcal D_L} = \min \{ s: \tilde{\bm Z} (s) \notin \mathcal D_L\}
\end{equation}
Now we define the limiting 
process $(\mathcal Z (s))_{0 \leq s \leq t}$ with $\mathcal Z (0) = x$
and with the generator
\[ \mathcal L = 
\begin{cases}
 \sum_{i=1}^d R(x)
\frac{\partial^2 }{\partial x_i^2}
+
 \sum_{i=1}^d \frac{\partial R}{\partial x_i }
\frac{\partial }{\partial x_i } & \text{in Theorem \ref{thm1}}\\
\sum_{i=1}^d r
\frac{\partial^2 }{\partial x_i^2} & \text{in Theorem \ref{thm2}(a)}\\
\sum_{i=1}^d r_{sign(x_1)}
\frac{\partial^2 }{\partial x_i^2} & \text{in Theorem \ref{thm2}(b)}.
\end{cases}
\]
acting on 
\begin{enumerate}
 \item functions $\phi \in \mathcal C^2_0$ in case of Theorem \ref{thm1}
 \item compactly supported continuous functions $\phi$ which
admit $\mathcal C^2$ extensions on $ (\mathbb R_- \cup 0) \times \mathbb R^{d-1}$
and $ (\mathbb R_+ \cup 0) \times \mathbb R^{d-1}$ and 
\[ 
\begin{cases}
\omega_{-1} \frac{\partial}{\partial x_1 -} \phi(0,x_2, ..., x_d)
= \omega_1 \frac{\partial}{\partial x_1 +} \phi(0,x_2, ..., x_d) & \text{
in case of Theorem \ref{thm2}(a)}\\
r_{-1} \frac{\partial}{\partial x_1 -} \phi(0,x_2, ..., x_d)
= r_1 \frac{\partial}{\partial x_1 +} \phi(0,x_2, ..., x_d) &
\text{ in case of Theorem \ref{thm2}(b)}
\end{cases}
\]
\end{enumerate}
Note that $\mathcal Z $ is a diffusion process in case of Theorem \ref{thm1} and a {\it generalized
diffusion process} in case of Theorem \ref{thm2} (see \cite{P90} for a survey on generalized diffusion 
processes and \cite{A78} for an early proof of the existence
of $\mathcal Z$ in case of Theorem \ref{thm2}(b)). Now let us stop $\mathcal Z$ on $\partial \mathcal D$ and define
\begin{equation}
\label{eq:YZ2}
\mathcal Y (s) = \mathcal Z (s \wedge \tau_{\mathcal D}) \quad \text{where} \quad 
\tau_{\mathcal D} = \min \{ s: \mathcal Z(s) \in \partial \mathcal D\}
\end{equation}

The connection between the processes $\tilde{\bm Y}$ and the PDE's defining $u$ is most easily
seen in the simplest case of one particle. 

Let $\Rightarrow$ denote weak convergence in the
Skorokhod space $\mathcal D[0,T]$ 
with respect to the supremum metric and with some $T$ to be specified.
(Although we need the Skorokhod space as the trajectories of ${\tilde{\bm Z} }$ are not continuous,
but the limiting measures will always be supported on $\mathcal C[0,T]$ and we can use the supremum metric. 
Alternatively, one could smooth the trajectories of ${\tilde{\bm Z}} $ and only use 
the space $\mathcal C[0,T]$.)

\begin{lemma}
\label{lemmaN=1,z}
If $N=1$, then 
$$ \left( \frac{\tilde{\bm Z} ({sL^2})}{L} \right)_{0 \leq s \leq t} \Rightarrow
 \mathcal (\mathcal Z (s))_{0 \leq s \leq t}.$$
\end{lemma}

\begin{proof}
In the setup of Theorem \ref{thm1}, this follows from Theorem 11.2.3
in \cite{SV07}.
More precisely, as we can neglect events of small probability, we can assume that 
(A) the particle jumps less than $L^3$ times before $tL^2$ and consequently
(B) the smallest time between two consecutive jumps before $tL^2$ is bigger than $L^{-4}$. Now choosing
$h=L^{-4}$, $\tilde{\bm Z}$ can only jump at most once on the interval $[kh, (k+1)h]$
for all $k <t L^6$. Now we choose 
\[
\Pi_h\left( \frac{z}{L}, \frac{z+e_i}{L}\right) = L^{-2} R \left( \frac{x+e_i/2}{L} \right), 
\quad \Pi_h\left( \frac{z}{L}, \frac{z}{L}\right) = 1-L^{-2} \sum_{e_i}
R \left( \frac{x+e_i/2}{L} \right)  
\]
for any $z \in \mathbb Z^d$ and any unit vector $e_i$. With this choice, we easily see that $a^{ii}(x) = 2R(x)$ and 
$b^i(x) = \frac{\partial}{\partial x_i} R(x)$ at the end of page 267 in \cite{SV07}. Applying 
Theorem 11.2.3, the Lemma follows.

In case of Theorem \ref{thm2}(a), we consider the first coordinate of $\tilde{\bm Z}$ and the other
$d-1$ coordinates separately. Under diffusive scaling, the former one converges to a skew Brownian
motion by e.g. \cite{ChShY04}, while the latter one converges to a $d-1$ dimensional Brownian motion by 
Donsker's theorem.
A slight technical detail is that we need to switch to discrete time so as to apply the result of \cite{ChShY04}.
Let us thus define ${ \bm Z}'_1 (k) = \tilde{ \bm Z}_1 (\tau_{1,k})$
for non-negative integers $k$, where $\tau_{1,0} = 0$
and $\tau_{1,k} = \min \{ t> \tau_{1, k-1}: \tilde{ \bm Z}_1 (t-) \neq  \tilde{ \bm Z}_1 (t+)\}$.  Then
by \cite{ChShY04} we have that
\begin{equation}
\label{eq:SBMd=1}
\left( \frac{{ \bm Z}'_1 (sL^2)}{L} \right)_{0 \leq s \leq 4rt}
\Rightarrow \left( \mathcal SBM_{\omega_1/(\omega_{-1}+ \omega_1)} (s) \right)_{0 \leq s \leq 4rt},
\end{equation}
where $SBM_{\omega_1/(\omega_{-1}+ \omega_1)}$ is the skew Brownian motion
with parameter $\omega_1/(\omega_{-1}+ \omega_1)$, which is by definition
equal to $\mathcal Z_1(s)/\sqrt{2r}$.
Now observe that by the law of large numbers, the functions 
$s \mapsto 2r\tau_{1,sL^2}/(sL^2)$
converge in probability to the identity function on $\mathcal D[0,t]$. Hence the statement of the Lemma, 
restricted to the first component, follows. Finally, the other components are independent from the first one
and under diffusive scaling, they converge to Brownian motion by Donsker's theorem. 

In case of Theorem \ref{thm2}(b) we use a similar argument to the one in the previous paragraph.
Namely, we still have the analogue of (\ref{eq:SBMd=1}) and by independence
\begin{equation}
\label{eq:SBMd}
\left(\frac{{ \bm Z}' (sL^2)}{L} \right)_{0 \leq s \leq T}
\Rightarrow
\left(  
 \mathcal SBM_{r_1/(r_{-1}+ r_1)} (s/d),
\mathcal W_{d-1} (s/d) \right)_{0 \leq s \leq T},
\end{equation}
where 
\[T=4 d \max \{ r_{-1}, r_{1}\} t,\] 
$\mathcal W_{d-1}$ is a $d-1$ dimensional standard Brownian motion and
${ \bm Z}' (k) = \tilde{ \bm Z} (\tau_{k})$ 
with $\tau_{0} = 0$ 
and $\tau_{k} = \min \{ t> \tau_{ k-1}:  \tilde{ \bm Z} (t-) \neq  \tilde{ \bm Z} (t+)\}$
for positive integers $k$. 
The difference from the case of Theorem \ref{thm2}(a) is that
now in order to recover the convergence of $\tilde{\bm Z}$, 
we need a nonlinear time change (and consequently have to work with all coordinates
at the same time). To do so, let us introduce the local time the first 
coordinate spends on the negative and positive halfline:
\[  
\tau_{-}(s) = \int_0^s 1_{f_1(y) < 0} dy, \quad
\tau_{+}(s) = \int_0^s 1_{f_1(y) > 0} dy,
\]
for $f = (f_1, ..., f_d) \in \mathcal D([0,T], \mathbb R^d)$. Now we can define
\[
 \rho(s) = \min \{ y: \frac{\tau_-(y)}{2 d r_-} + \frac{\tau_+ (y)}{2 d r_+} \geq s\}
\]
and
\[
 \Theta: \mathcal D([0,T]) \rightarrow \mathcal D([0,3t/2])
\]
by
\[
(\Theta (f)) (s) =
\begin{cases}
 f(\rho(s)) & \text{ if } \tau_-(T) + \tau_+(T) \geq 3T/4 \\
0 & \text{ otherwise. }
\end{cases}
\]
The function $\Theta$ is well defined because of the choice of $T$.
Let us denote by $Q_L$ and $Q$ the measure on $\mathcal D[0,T]$ given by the left and right hand sides of
(\ref{eq:SBMd}). Next, we observe that
\begin{equation*}
Q(\tau_-(T) + \tau_+(T) = T) = 1.
\end{equation*}
and $\Theta$ is continuous on the set $\{ \tau_-(T) + \tau_+(T) = T \}$.
Now we can apply the continuous mapping theorem (see e.g. Theorem 5.1 in \cite{B68})
to conclude 
$\Theta_*Q_L \Rightarrow \Theta_* Q$. Here, $\Theta_*Q_L $ is the distribution of a process 
$(\hat{\bm Z}(s))_{0 \leq s \leq 3t/2}$ obtained from $({\bm Z}'(s))_{0 \leq s \leq T}$ 
by rescaling time with $\Theta$.
Since $\Theta_* Q$ is the distribution of $\mathcal Z(s)$, it suffices to show that 
\begin{equation}
\label{eq:AA}
 \left( \frac{\tilde{\bm Z} ({sL^2}) - \hat{\bm Z} ({sL^2})}{L} \right)_{0 \leq s \leq t}
\Rightarrow 0.
\end{equation}

Let us introduce the notations $z_k = {\bm Z}' ({k})  - {\bm Z}' ({k-1})$, $\mathcal T_i = \tau_{i} - \tau_{i-1}$.
By an elementary estimate on the local time of random walks, we have
$$ \mathbb P (\# \{ k< tL^2: \bm Z' (k)_1 =0\} > L^{3/2}) = o(1). $$
Since we prove weak convergence, we can neglect the above event and subsequently assume that
$z_1, ..., z_{tL^2}$ is such that the first coordinate's local time at zero is smaller than 
$L^{3/4}$. Especially, we use the first line of the definition of $\Theta$ when we construct $\hat{\bm Z}$
for $L$ large. Thus with the notation
$$\tilde i_s = \max \{ i: \sum_{j=1}^i \mathcal T_j < s \} \quad
\hat i_s = \max \{ i: \sum_{j=1}^i \mathbb E (\mathcal T_j 1_{\{ {\bm Z}' (j)_1 \neq 0 \}} | z_1, ..., {z_{tL^2}}) 
< s \}.$$
we have $\tilde{\bm Z} (s) = {\bm Z}' (\tilde i_s)$
and $\hat{\bm Z} (s) = {\bm Z}' (\hat i_s)$.
Now with some fixed small $\delta$ let us write 
$$ \sum^k = \sum_{i\in [k \delta L^2, (k+1)\delta L^2]} \quad \text{and} \quad
\sum^{k*} = \sum_{i\in [k \delta L^2, (k+1)\delta L^2], {\bm Z} (i)_1 \neq 0}$$
for $k=1,2,..., T/\delta$. Then by Chebyshev's inequality,
\begin{eqnarray}
&&\mathbb P \left( \left|\sum^{k*} \mathcal T_i - \sum^{k*} \mathbb E (\mathcal T_i | z_1, ..., {z_{tL^2}})\right|
>  \delta^3 L^2 | z_1, ..., z_{z_{tL^2}} \right) \nonumber \\
&<& \frac{\delta L^2}
{4d^2 \min \{r_1^2, r_{-1}^2 \}} \frac{1}{\delta^6 L^4} < \delta^{10} \nonumber
\end{eqnarray}
for $L$ large enough.
Furthermore, by our assumption on $z_1, ..., {z_{tL^2}}$, 
\[ 
P \left( \left|\sum^{k*} \mathcal T_i  - \sum^{k} \mathcal T_i \right| > \delta^3 L^2
 | z_1, ..., {z_{tL^2}}  \right)
 < \delta^{10}
\]
also holds for $L$ large. Consequently the random times 
$\tilde s_l = \sum_{k=1}^l \sum^{k*} \mathbb E (\mathcal T_i | z_1, ..., {z_{tL^2}})$
and $\hat s_l = \sum_{i=1}^{lL^2/\delta} \mathcal T_i$
satisfy
\begin{itemize}
 \item $0=\tilde s_0<\tilde s_1<\tilde s_2 ...< \tilde s_T$, $\tilde s_T \geq 3t/2$ and
$\tilde s_l - \tilde s_{l-1} < C\delta$
 \item for all $l$ with $\tilde s_l < t$,
$\mathbb P (|\hat s_l - \tilde s_l| > 2 \delta^2 ) < \delta^{8}$ and 
$\hat{\bm Z} (\hat s_l L^2) = \tilde{\bm Z} (\tilde s_l L^2)$
\end{itemize}
This together with tightness of $\bm Z'$ (proved in the usual way) gives (\ref{eq:AA}). We have finished the proof of 
Lemma \ref{lemmaN=1,z}.
\end{proof}

\begin{lemma}
\label{lemmaN=1,y}
If $N=1$, then 
$$ \left( \frac{\tilde{\bm Y}^{(L)}({sL^2})}{ L} \right)_{0 \leq s \leq t} \Rightarrow
 \mathcal (\mathcal Y (s))_{0 \leq s \leq t}.$$
\end{lemma}

\begin{proof}
 This is a consequence of the continuous mapping theorem, Lemma \ref{lemmaN=1,z},
(\ref{eq:YZ1}) and (\ref{eq:YZ2}).
\end{proof}

Next, we prove a simpler version of Theorems \ref{thm1} and \ref{thm2}, namely the convergence of the
expectations

\begin{proposition}
\label{prop:exp}
Consider the setup of either Theorem \ref{thm1} or Theorem \ref{thm2}. Then
\begin{equation}
\label{eq:moments1}
 \lim_{L \rightarrow \infty} \mathbb E \left( \xi_{\langle xL \rangle}^{(L)} (t L^2)\right) = 
u(t,x)
\frac{\omega_{x }}{2},
\end{equation}
where $\omega_x = \omega$ in case of Theorems \ref{thm1} and \ref{thm2}(b) 
and $=\omega_{sign(x_1)}$ in case of Theorem \ref{thm2}(a)
\end{proposition}

\begin{proof}
By Proposition \ref{prop:dual}, the left hand side of (\ref{eq:moments1}) is equal to
\begin{equation*}
\frac{\omega_{\langle xL \rangle }}{2}
\mathbb E  
\left( \xi^*_{\tilde{\bm Y} (tL^2)} \frac{2}{\omega_{\tilde{\bm Y} (tL^2)}} 
1_{ \{ \tilde{\bm Y} (tL^2)  \in \mathcal D_L \}} +
T \left(\frac{{\tilde{\bm Y}({tL^2})}}{L} \right) 1_{ \{\tilde{\bm Y}({tL^2}) \in \mathcal B_L \}} \right)
=:I + II ,
\end{equation*}
where $\xi^*_v = \xi_v(0)$.
In case of Theorems \ref{thm1} and \ref{thm2}(b), $\omega$ is constant, thus we obtain
$$
I = \sum_{v \in \mathcal D_L} \mathbb P ( \tilde{\bm Y} (tL^2)=v ) \mathbb E (\xi_v^*).
$$
Recall that (\ref{eq:init}) gives $\mathbb E (\xi_v^*) = \frac{\omega}{2} f(v/L) + o(1)$. Applying
Lemma \ref{lemmaN=1,y} and the
definition of the weak convergence, we obtain
$$
\lim_L I = \frac{\omega}{2} \mathbb E (f(\mathcal Y (t)) 1_{\{\mathcal Y(t) \in \mathcal D\}}) 
$$ 
Similarly,
$$
\lim_L II= \frac{\omega}{2} \mathbb E (T(\mathcal Y (t)) 1_{\{\mathcal Y(t) \in \mathcal \partial D\}}) 
$$ 
Now the proposition follows from the fact that $u(t,x)$ solves the Kolmogorov equation associated with the 
process $\mathcal Y$. 
Finally, in case of Theorem \ref{thm2}(a), we 
consider the function $h$ on $ \mathcal D $ with $h(y) = 0$ if $y_1 <0$, $h(y) = y_1/\delta$ if
$0 \leq y_1 < \delta$ and $h(y) = 1$ if $y_1 > \delta$ and approximate $I$ by $Ia + Ib$, where
$Ia$ is obtained from $I$ by multiplying the integrand with $h(\tilde{\bm Y} (tL^2)/L)$
and $Ib$ is obtained from $I$ by multiplying the integrand with $h( - \tilde{\bm Y} (tL^2)/L)$. Clearly,
$I - \varepsilon < Ia + Ib < I$ for $\delta$ small enough. Then we can repeat the above argument since 
$\omega$ is constant on the integration domain in both $Ia$ and $Ib$.
\end{proof}

In order to complete the proof of Theorems \ref{thm1} and \ref{thm2}, we need an extension of Lemma \ref{lemmaN=1,y}
from $N=1$ to arbitrary $N$:
 
\begin{proposition}
\label{lemmaN,y} 
The processes
$$ \left( \frac{\tilde{\bm Y}_i^{(L)}({sL^2})}{ L} \right)_{0 \leq s \leq t} $$
for $i=1,...,N$ converge weakly to $N$ independent copies of $(\mathcal Y(s))_{0 \leq s \leq t}$.
\end{proposition}

Now, we prove Theorems \ref{thm1} and \ref{thm2} assuming Proposition \ref{lemmaN,y}.
By the discussion at the beginning of this section
and by (\ref{eq:gammamoment}), Theorems \ref{thm1} and \ref{thm2} will be proved
once we establish
\begin{equation}
\label{eq:moments}
\mathbb E \left( \prod_{s \in S} \xi_{\langle xL \rangle +s}^{n^*_s}(tL^2) \right) \sim
u(t,x)^{N} \prod_{s \in S} \frac{ \Gamma(n^*_s + \omega_{\langle xL \rangle +s}/2)}
{\Gamma(\omega_{\langle xL \rangle +s}/2)}.
\end{equation}
By Proposition \ref{lemmaN,y}, 
$\lim_{L} \mathbb P (\mathcal A^{(L)}_{\delta}) =  o_{\delta} (1)$, where
$$ 
\mathcal  A_{\delta} = \mathcal  A^{(L)}_{\delta} = \{ \| {\tilde{\bm Y}_i^{(L)}({tL^2})} - {\tilde{\bm Y}_j^{(L)}({tL^2})} \| 
\leq \delta L
\text{ for some } i \neq j\}.
$$
This, combined with the uniform moment condition and the Cauchy-Schwarz inequality gives
$$
\lim_{L}  \mathbb E \left( 1_{\mathcal A_{\delta}} F(\bm Y (t L^2), \underline \xi^*)
\right)  = o_{\delta} (1),
$$
where $\xi^*_v = \xi_v(0)$ and $F$ is the duality function defined in (\ref{eq:F}).
Now Proposition \ref{prop:dual} and the fact that $\underline \xi^*$ is associated with $f$
implies that the left hand side of (\ref{eq:moments}) is $o_{\delta} (1)$ close to
\begin{eqnarray*}
&&\prod_{s \in S} \frac{ \Gamma(n^*_s + \omega_{\langle xL \rangle +s}/2)} 
 {\Gamma(\omega_{\langle xL \rangle +s}/2)}\times \\
&\times &\mathbb E \left[ 1_{\overline{ \mathcal A}_{\delta}} \prod_{i=1}^N 
\left( \xi^*_{\tilde{\bm Y} (tL^2)} \frac{2}{\omega_{\tilde{\bm Y} (tL^2)}} 
1_{ \{ \tilde{\bm Y} (tL^2)  \in \mathcal D_L \}} +
T \left(\frac{{\tilde{\bm Y}({tL^2})}}{L} \right) 1_{ \{\tilde{\bm Y}({tL^2}) \in \mathcal B_L \}} \right) \right].
\end{eqnarray*}
Now we can cut this integral to $2^N$ pieces and apply a version of the proof of 
Proposition \ref{prop:exp} to conclude (\ref{eq:moments}).

In order to complete the proof of Theorems \ref{thm1} and \ref{thm2} it only remains to prove Proposition 
\ref{lemmaN,y}, which is the subject of the next section.

\section{Proof of Proposition \ref{lemmaN,y}}
\label{sec6}

We are going to prove a variant of Proposition \ref{lemmaN,y} obtained by replacing $\tilde{ \bm Y}_i$ and
$\mathcal Y$
with $\tilde{\bm Z}$ and $\mathcal Z$. Proposition \ref{lemmaN,y} follows from this variant the same way
as Lemma \ref{lemmaN=1,y} follows from Lemma \ref{lemmaN=1,z}.

The idea of the proof is borrowed from \cite{LNY15}, Section 5: we show that with probability close to $1$, 
$\tilde{\bm Z}_i(s)$ and $\tilde{\bm Z}_j(s)$ will not meet after getting separated by a distance $L^{\gamma}$
with some $\gamma$ close to $1$. Since $\tilde{\bm Z}_i(s)$ and $\tilde{\bm Z}_j(s)$
move independently if their distance is bigger than $1$, we 
can replace $\tilde{\bm Z}(s)$ for 
$s > \tau = \max \{ \tau_{i,j}\}$ 
by $N$ independent copies of $\tilde{\bm Z}_1(s)$, where $\tau_{i,j}$ is the first time
$s$ when $\|\tilde{\bm Z}_i(s) - \tilde{\bm Z}_j(s) \| >  L^{\gamma}$. It only remains to show that 
$\tilde{\bm Z}_i(\tau)$ is close to $xL$ and $\tau/L^2$ is negligible. We complete this strategy
in the case of Theorem \ref{thm1}, 
$N=2$ and $d=2$ in Section \ref{sec6.1} and for all other cases in Section \ref{subsec:6.3}.

\subsection{Case of Theorem \ref{thm1}, 
$N=2$ and $d=2$.}
\label{sec6.1}

Recalling the notation of Section \ref{sec:5}, let us write $Z(k) = {\bm Z}_2'(k) - {\bm Z}_1'(k)$.
Note that $ Z$ is not a particularly nice process: it is neither Markov, nor translation invariant.
As long as both ${\bm Z}_1'(k)$ and ${\bm Z}_2'(s)$ are $\varepsilon L$-close to 
$xL$, $ Z$ is well approximated by a Brownian motion. In particular, if $\| Z\| = M$, then
the probability that $ \|  Z\|$ reaches $M/2$ before reaching $2M$ is close to $1/2$ and thus
$\| Z\|$ performs an approximate simple symmetric random walk (SSRW) on the circles
$$ \mathcal C_m = \{ z \in \mathbb Z^2: | \|z \| - 2^m | <2\}.$$
We will estimate the goodness of this approximation by a SSRW.

Assuming that $ Z(s_0) \in \mathcal C_m$, denote by $s_1$ the smallest $s > s_0$ such that 
$Z(s) \in \mathcal C_{m-1}$ or $Z(s) \in \mathcal C_{m+1}$. We also write 
$$\log = \log_2, \quad M=2^m, 
\quad \mathcal E_{L,M} = \left| \mathbb P (Z(s_1) \in \mathcal C_{m-1}) - \frac12 \right|. $$ 
Note that $s_1$ is a stopping time with respect to the filtration generated by
$ {\bm Z}_2'(s), {\bm Z}_1'(s)$.

We will need a series of lemmas.

\begin{lemma}
\label{lemma:coupleRW}
If $\| Z(k)\| \geq 2 $, then
\begin{equation*}
\mathbb P \left( 
Z(k+1) - Z(k) = e 
\right) = \frac{1}{4} + 
O\left( \frac{ \| Z(k) \|}{L^2} \right)
\end{equation*}
for $e=(1,0), (-1,0),(0,1),(0,-1)$. Here, the constants involved in $O$ only depend on the
$\mathcal C^2$ norm of $R$.
\end{lemma}

\begin{proof}
Using that $R$ is a $\mathcal C^2$ function on a compact domain containing $\mathcal D$,
the lemma follows from the definition.
\end{proof}

Lemma \ref{lemma:coupleRW} enables us to couple $Z(k)$ with a planar SSRW $W(k)$ such that 
\[
\mathbb P (Z(k) - Z(k-1) = W(k) - W(k-1) | Z(k) ) \geq 1-\frac{C \| Z(k)\|}{L^2}.
\]
We are going to apply such a coupling several times in the forthcoming lemmas.

\begin{lemma}
\label{lemma:wasp}
There are constants $\varepsilon_0 >0$, 
 $C_0< \infty$ and $\theta<1$ such that 
\begin{equation*}
\mathbb P (s_1 > n) < C_0 \theta^{n/ M^2}
\end{equation*}
assuming $\| Z(s_0)\| = M < \varepsilon_0 L$.

\end{lemma}

\begin{proof}
To verify Lemma \ref{lemma:wasp}, we couple 
$Z(k), k \in [s_0, s_0 + M^2]$ to a SSRW $W(k),  k \in [s_0, s_0 + M^2], W(s_0) = Z(s_0)$.
By Lemma \ref{lemma:coupleRW}, we can guarantee that $\| Z(k) - W(k) \| < C M^3 / L^2$ for all $k < M^2$
with some probability bounded away from zero. 
Here, $C$ only depends on the $\mathcal C^2$ norm of $R$. Thus by choosing
$\varepsilon_0$ small, we can guarantee $C M^3 / L^2 < M/10$.
Since $W$ is close to a Brownian motion, it reaches 
$\mathcal C_{m-2}$ or $\mathcal C_{m+2}$ before $M^2$ with some positive probability.
Consequently, there is some 
$p$ independent of $L$ such that $\mathbb P (s_1 < M^2) > p$. Applying this argument in an inductive fashion
gives  Lemma \ref{lemma:wasp}.
\end{proof}

Let us introduce the notation
\begin{equation}
\label{eq:deftau}
\underline \tau_{\gamma} =
\underline \tau_{\gamma} (L)= \min \{ k: \| Z(k) \| > L^{\gamma}\}
\end{equation}
Now we claim the following

\begin{lemma}
\label{lemma:annuli}
For every $\gamma<1$ with $1-\gamma$ small, there exists some $\xi >0$ such that  
for $L$ large enough, the following estimates hold.
\begin{enumerate}
 \item[(a)] Tiny gap:  If $m < \frac{3}{5} \log L$, then 
\[
 \mathbb P ( \tau_{3/5} > L^{13/10} ) = O(L^{-1/10})
\]
 \item[(b)] Small gap: 
 $$ \text{ If $\frac{3}{10}  \log L \leq m < \gamma \log  L$, then
 $\mathcal E_{L,M} = O(L^{-\xi})$} $$
 \item[(c)] Moderate gap: 
 $$ \text{ If $\gamma \log L < m < \log L - \log \log L$, then 
 $\mathcal E_{L,M} = O(\log^{-3/2} L)$ }$$
 \item[(d)] Large gap: There is some $\varepsilon >0$ such that 
 $$ \text{ if $\log L - \log \log L < m < \log L + \log \varepsilon$, then
   $\mathcal E_{L,M} < 1/100$ }$$
\end{enumerate}
\end{lemma}

As the proof of Lemma \ref{lemma:annuli}
is slightly longer than the other lemmas, we postpone it to the Appendix.
Next, we formulate our key lemma:

\begin{lemma}
\label{lemma:alpha}
For any small $\delta>0$, there exists $\gamma<1$ such that for $L$ large enough,
\[ \mathbb P ( \nexists k: \underline \tau_{\gamma} < k < tL^2:  \| Z(k) \| \leq 2)
 > 1-\delta.
\]
\end{lemma}

\begin{proof}
In order to derive Lemma \ref{lemma:alpha} from Lemma \ref{lemma:annuli}, recall the connection
of random walks and electrical networks from Section \ref{sec:pr}. Using the notation from there, 
we choose $A=\frac{3}{10} \log L$, $B = \log L + \log \varepsilon$ and $I = \gamma \log L$,
$R_{I+1/2} =1$. 
If $R_{i+1/2}$ is defined for some $i \in [I, \log L - \log \log L]$, then let us define
$R_{i+3/2} = R_{i+1/2}(1+ K \log^{-3/2} L )$. If 
$R_{i+1/2}$ is defined for some $i \in [\log L - \log \log L, B-1]$, then let us define 
$R_{i+3/2} = \frac{11}{10} w_{i+1/2}   $. Similarly, if $R_{i+1/2}$ is defined for some $i \in [A,I]$,
then we define $R_{i-1/2} = R_{i+1/2}(1- K L^{-\xi} )$. 
Now by Lemma \ref{lemma:annuli}(b-d)
$$
\mathbb P (\min \{k: \| Z (k+s_0) \| < L^{3/10} \} < \min \{k:  \| Z (k+s_0) \| > \varepsilon L \} 
| \| Z(s_0)\| = L^{\gamma}) 
$$
is bounded from above by (\ref{eq:electric}). An elementary computation shows that (\ref{eq:electric}) can 
be made arbitrarily small by choosing $\gamma $ close to $1$. Thus 
after $\underline \tau_{\gamma}$, $ \| Z(k)\| $ reaches $\varepsilon L$ before reaching $2$ with probability
close to $1$. Finally, Lemmas \ref{lemma:subsyst} and \ref{lemmaN=1,z} yield that for fixed $\varepsilon $
the two particles do not meet after separating by a distance $\varepsilon L$ and before $tL^2$ with 
probability close to $1$.
This proves Lemma \ref{lemma:alpha}.
\end{proof}

As a consequence of Lemma \ref{lemma:alpha}, we will be able to replace $\tilde{\bm Z}_1(s)$ and 
$ \tilde{\bm Z}_2(s)$ with two independent copies after time 
\[ \underline{ \tilde \tau}_{\gamma} = \min \{ s: \| \tilde{\bm Z}_1(s) - \tilde{\bm Z}_2(s) \| > L^{\gamma}\}. \]
Then 
Proposition \ref{lemmaN,y} will easily follow once we establish that (A) $\underline{\tilde \tau}_{\gamma} / L^2$
is negligible and (B)
$(\tilde{\bm Z}_i( \underline \tau_{\gamma}) - xL)/L$ is negligible. This is what we do in the next two lemmas.

\begin{lemma}
\label{lemma:alpha2a}
For any fixed $\gamma <1$ and $\delta >0$, we have
\[
\mathbb P(\underline \tau_{\gamma} < L^{1+\gamma}) > 1-\delta
\]
for $L$ large enough.
\end{lemma}

\begin{proof}
 By Lemma \ref{lemma:annuli}(a), it is enough to prove 
$\mathbb P(\underline \tau_{\alpha} - \underline \tau_{3/5} > L^{1+\gamma}) < \delta$.
We write $s_0 = \underline \tau_{3/5} $ and if $Z(s_{i}) \in \mathcal C_{m}$ with $m > \frac{3}{10} \log L$, 
then $s_{i+1}$ is the 
first time $s$ when either $Z(s) \in \mathcal C_{m-1}$ or $Z(s) \in \mathcal C_{m+1}$.
If $Z(s_{i}) \in \mathcal C_{\lfloor \frac{3}{10} \log L \rfloor}$, then 
$s_{i+1}$ is the 
first time $s$ when $Z(s) \in \mathcal C_{\lfloor \frac{3}{10} \log L \rfloor +1}$.
By Lemma \ref{lemma:annuli}(b), $b_i := \log \| Z(s_i)\|$ can be approximated by a one dimensional SSRW 
(reflected at $\frac{3}{10} \log L$, absorbed at $\gamma \log L$) with an error
of $O(L^{\xi})$ at each step. 
In particular, if $\bm t$ is the smallest $i$ when $b_i \geq \gamma \log L$, then 
$\mathbb P (\bm t > \log^3 L) < \delta /10$ and 
thus $b_i, i \leq \bm t$ can be coupled to a SSRW with an error $< \delta /5$.
Now if 
$\zeta = (1-\gamma)/2$ and  
$\ell_m = \# \{ i < \bm t : b_i = m\}$, then
\[
\mathbb P (\exists m: \ell_m > L^{\zeta}) \leq \log L \max_m \mathbb P (\ell_m > L^{\zeta}) < \delta/10 
\]  
by the gambler's ruin estimate $\mathbb P(\ell_m > n+1 | \ell_m >n) < 1-\log^{-1} L$.
If $\ell_m < L^{\zeta}$ for all $m$, then
\[
\underline{\tau}_{\alpha} < \sum_{m = \frac{3}{10} \log L}^{\gamma \log L} \sum_{i=1}^{L^{\zeta}}
 \mathcal T_{m,i},
\]
where $\mathcal T_{m,1}$ is the random time $s_1 - s_0$ if $\| Z(s_0)\| = m$ and for fixed $m$, 
$\mathcal T_{m,i}$'s are iid. Consequently, we have 
\begin{eqnarray*}
\mathbb P(\underline \tau_{\alpha} > L^{1+\gamma}) < \frac{3\delta}{10} + 
(\log L ){L^{\zeta}} \max_m
\mathbb P \left( 
 \mathcal T_{m,1}> L^{1+\gamma - \zeta} \log^{-1} L 
\right) < \delta
\end{eqnarray*}
by Lemma \ref{lemma:wasp} and Lemma \ref{lemma:annuli}(a).
\end{proof}

Now, with the notation
\[
\tilde{\overline \tau}_{\gamma} = \tilde {\overline \tau}_{\gamma}( L)= 
\min \{ k: \| \tilde{\bm Z}_1(k) - xL \| > L^{\gamma} \text{ or }
\| \tilde{\bm Z}_2(k) - xL \| > L^{\gamma }\}
\]
we have

\begin{lemma}
\label{lemma:alpha2}
For any fixed $\gamma <1$ and $\delta >0$, we have
\[
\mathbb P(  L^{1+ \gamma} < \tilde{ \overline \tau}_{\frac{\gamma + 3}{4}}  ) > 1- \delta
\]
for $L$ large enough.
\end{lemma}

\begin{proof}
By Lemma \ref{lemma:subsyst}, it is enough to consider the case of one particle.
A simplified version of 
Lemma \ref{lemmaN=1,z} yields that 
\[ \max_{s < L^{1+\gamma} } \| \tilde{\bm Z}_1 (s) -xL \| < K L^{\frac{1+\gamma}{2}}
\]
with probability $1 - \delta $ for some $K$ and $L$ large.
\end{proof}

The variant of Proposition \ref{lemmaN,y} (explained in the beginning of Section \ref{sec6})
easily follows from Lemmas \ref{lemma:alpha}, \ref{lemma:alpha2a}
and \ref{lemma:alpha2}. Since we can neglect an event of
small probability, we can assume that all events hold in 
Lemmas \ref{lemma:alpha}, \ref{lemma:alpha2a} and \ref{lemma:alpha2}. 
Note that by definition,
$\tilde{\bm Z}_1(s)$ and  $ \tilde{\bm Z}_2(s)$ move independently if
their distance is bigger than $2$. Thus for $s > \underline{ \tilde \tau}_{\gamma}$, we can replace 
$(\tilde{\bm Z}_1( \underline{ \tilde \tau}_{\gamma} +k))_{k=1,...,TL^2}$ 
and $( \tilde{\bm Z}_2(\underline{\tilde \tau}_{\gamma} + k))_{k=1,...,TL^2}$ 
with independent random walks, both of them converging
to $\mathcal Z (s)$ under the proper scaling. Finally, by Lemmas \ref{lemma:alpha2a}
and \ref{lemma:alpha2},
$\underline \tau_{\gamma}/L^2 < \delta$ and 
$(\tilde{\bm Z}_i( \underline{\tilde \tau}_{\gamma}) - x)/L < \delta$. Proposition 
\ref{lemmaN,y} follows in the case of Theorem \ref{thm1}, $N=2$, $d=2$.

\subsection{Completing the proof of Proposition \ref{lemmaN,y}}
\label{subsec:6.3}
The case of general $N$
follows from Lemma \ref{lemma:subsyst} and from the case $N=2$. Indeed, Lemmas \ref{lemma:subsyst},
\ref{lemma:alpha} and \ref{lemma:alpha2a}
imply that for some $\gamma < 1$,
\[
\mathbb P (\exists s \in [L^{1+\gamma}, tL^2], i,j \in \{ 1,2,...,N\}: 
\| {\tilde{\bm Z}}_i(s) - {\tilde{\bm Z}}_j(s)\| \leq 2 ) < \delta.
\]
Furthermore, we have the analogue of Lemma \ref{lemma:alpha2} with 
$\tilde{\overline \tau}_{\alpha} $ replaced by
\[
\min \{ k: \exists i < N: \| \tilde{\bm Z}_i(k) - xL \| > L^{\alpha} \}.
\]
Whence the case of general $N$ follows the same way as before.\\

The case of dimension $d > 2$ is simpler than $d=2$ as the particles only
meet finitely many times.
\begin{lemma}
\label{lemma:d>2}
In case of Theorem \ref{thm1}, $d>2$, $N=2$, there is some positive $p$ such that 
$\| {\tilde{\bm Z}}_1(s) - {\tilde{\bm Z}}_2(s)\| \geq 2$ for all $k \in [2, tL^2]$ with probability at least $p$.
\end{lemma}

\begin{proof}
Consider the process $Z(k) = {\bm Z}'_2(k) - {\bm Z}'_1(k)$ as in Section \ref{sec6.1}.
Let us fix some small $\varepsilon>0$. We will show that the probability of the event 
$ \{ \| Z \|$ reaches $\varepsilon L$ before reaching $1 \}$ is bounded away from zero. From 
this the lemma will follow by the same argument as in $d=2$ (cf. 
the end of the proof of Lemma \ref{lemma:alpha}).

Similarly to Lemma \ref{lemma:coupleRW}, we have 
\begin{equation}
\label{eq:d>2}
\frac{1}{2d} - \frac{c_0 \varepsilon}{ L} < \mathbb P \left( 
Z(k+1) - Z(k) = e 
\right) <
\frac{1}{2d} + \frac{c_0 \varepsilon}{ L}
\end{equation}
assuming $\| Z(k)\| \leq \varepsilon L$
for all unit vectors $e \in \mathbb Z^d$ and $L$ large enough.
Now we define a random walk $B(k)$ on $\mathbb Z^d$ with weights. 
Specifically $B(3) = Z(3)$ and we choose the weights
\[
w_{(u,v)} = \left( 1-\frac{ c_1 \varepsilon}{L} \right)^l, 
\text{ where $|u|_1 = l-1$, $|v|_1 = l$ and $c_1 \gg c_0$ is a fixed constant.} 
\]
Clearly $\| Z(3) \| \geq 2$ holds with some positive probability. 
Let us write $\bm t_{B,l} = \min \{ k > 3: | B(k)|_1 = l \}$, where $\min \emptyset = \infty$ and similarly
$\bm t_{ Z,l} = \min \{ k > 3: \| Z(k)\| = l \}$. 
Now we claim that 
\begin{lemma}
\label{sublemma}
Assuming $c_1 = c_1(c_0)$ is large enough, there exists a coupling
between the processes $B$ and $Z$ such that
$|B_i(k)| \leq |Z_i(k)|$ holds 
for all $k \in [3, \bm t_{B,1}  \wedge \bm t_{Z,\varepsilon L})$ and $i \leq d$ almost surely. 
\end{lemma}
By Lemma \ref{sublemma}, it suffices to prove that 
\begin{equation}
\label{eq:d>2,3}
\mathbb P(\bm t_{B, \varepsilon L} < \bm t_{B,1}) \text{ is bounded away from zero.}
\end{equation}
This follows from a simple application of the connection between random walks and
electrical networks. 
Note that the weights $w_{(u,v)}$ are bounded away from zero
in the $\varepsilon L$ neighborhood of the origin. In such cases, (\ref{eq:d>2,3}) follows from a standard
argument, see e.g. the proof of Theorem 19.30 in \cite{K14}. In order to complete the proof of the Lemma
\ref{lemma:d>2}, it only remains to prove Lemma \ref{sublemma}.
\end{proof}

\begin{proof}[Proof of Lemma \ref{sublemma}]
We prove by induction on $k$. Assume $|B_i(k)| \leq |Z_i(k)|$ for all $i \leq d$.

{\bf Case 1}: none of the coordinates of $B(k)$ and $Z(k)$ are zero.\\
If $|B_i(k+1)| = |B_i(k)|+1$, then we define $Z(k+1)$ by $|Z_i(k+1)| = |Z_i(k)|+1$. 
We can do so by (\ref{eq:d>2}) and by the definition of $B$: if $c_1$ is large enough, then
\begin{equation}
\label{eq:sub1}
\mathbb P( |B_i(k+1)| = |B_i(k)|+1 ) \leq \mathbb P( |Z_i(k+1)| = |Z_i(k)|+1 ).
\end{equation}
Similarly, if $|Z_i(k+1)| = |Z_i(k)|-1$, then we define $B(k+1)$ by $|B_i(k+1)| = |B_i(k)|-1$.
We can do so as  similarly to (\ref{eq:sub1}), we have
\[
\mathbb P( |Z_i(k+1)| = |Z_i(k)|-1 ) \leq \mathbb P( |B_i(k+1)| = |B_i(k)|+1 ).
\]
The coupling is arbitrary on the remaining set (sometimes we may have to move $B$ and $Z$ in 
different directions to match the probabilities, i.e. to define a proper coupling). 
This
proves the inductive step for the case when none of the coordinates of $B(k)$ and $Z(k)$ are zero. 

{\bf Case 2:}
none of the coordinates of $B(k)$ are zero or
for all $i$ with $B_i(k) = 0$, $Z_i(k) =0$ also holds.\\
The same argument works as in Case 1.

Note that the coupling of Case 1 will not work if
$B_i(k) =0$ and $Z_i(k) \neq 0$ as $\mathbb P( |B_i(k+1)| = |B_i(k)|+1 ) \approx d^{-1}$ 
while $\mathbb P( |Z_i(k+1)| = |Z_i(k)|+1 ) \approx (2d)^{-1}$.
Let $\mathcal I$ denote
the set of indices $i$ with $B_i(k) = 0, Z_i(k) \neq 0$ and let $\mathcal I' \subset \mathcal I$
be the set of indices $i$ with $B_i(k) = 0$ and $|Z_i(k)|=1$.

{\bf Case 3:} $|Z_i(k)| \geq 2$ for all $i \in \mathcal I$. \\
If $|B_i(k+1)| = |B_i(k)|+1$
with some $i \in \mathcal I$, then we can define
$Z(k+1)$ by changing the $i$th coordinate (either increasing or decreasing), 
since we have
\begin{eqnarray}
&&\mathbb P( B(k+1) = B(k)+e_i \text{ or }  B(k)-e_i) \nonumber \\
&\leq& \mathbb P( Z(k+1) = Z(k)+e_i \text{ or }  Z(k)-e_i). \label{eq:sub2}
\end{eqnarray}
The proof of (\ref{eq:sub2}) is similar to that of
(\ref{eq:sub1}): the first line of (\ref{eq:sub2}) is
\[
\frac{w_{(B(k),B(k)+e_i)} + w_{(B(k),B(k)-e_i)}}{\sum_{v: \| v-B(k)\| =1} w_{(B(k),v)}}.
\]
Here the denominator can be bounded from below by  
\[
{\left( 1- \frac{c_1 \varepsilon}{L} \right)^{|B(k)|_1} \left[1+(2d-1)\left(1- \frac{c_1 \varepsilon}{L} 
\right) \right]},
\]
which corresponds to the case when only one coordinate is nonzero
(note that we have excluded $B(k)=0$ since $k < \bm t_{B,1}$). 
Combining this estimate with (\ref{eq:d>2}) gives (\ref{eq:sub2}) assuming that
 $c_1=c_1(c_0)$ is large enough.
The other coordinates are treated the same way as in Case 1.

{\bf Case 4:} $|\mathcal I'| \neq 0$ is an even integer.

Consider a perfect matching of $\mathcal I'$. Let $(i,j) \in \mathcal I'^2$ be an arbitrary pair.
If $Z_i(k+1) = 0$, then we define $B(k+1)$ such that $|B_j(k+1)| = 1$. We can do so, since 
$\mathbb P (Z_i(k+1) = 0) \approx (2d)^{-1}$ and $\mathbb P (|B_j(k+1) |= 1) \approx d^{-1}$.
Analogously, if $Z_j(k+1) = 0$, then we define $B(k+1)$ such that $|B_i(k+1)| = 1$. 
Then,
we do the coupling of movements in other directions as discussed in cases 1-3. Finally,
we couple the remaining set (including $Z_i(k+1) = 2$, $Z_j(k+1) = 2$) arbitrarily.

{\bf Case 5:} $|\mathcal I'| =n$ is an odd integer.

First we consider a matching of $n-1$ elements of $\mathcal I'$, and do the coupling described in Case 4.
Let us denote the remaining index by $i$.
Now we claim that there is some $j \notin \mathcal I'$ such that $|Z_j(k)| -| B_j(k)| \geq 1$. Indeed, if there
was no such $j$, then $\sum_{m=1}^d |Z_m(k)| -| B_m(k)| = n$ would be an odd number, 
which is a contradiction with 
the fact that $Z(0) = B(0)$ and both processes move to nearest neighbors at each step.
Now if $|Z_j(k+1)| = |Z_j(k)|-1$, then we define $B(k+1)$ such that $|B_i(k+1)| = 1$. 
If $Z_i(k+1) = 0$, we define $B(k+1)$ by changing the $j$th coordinate (either decreasing or increasing).
These can be done as before.
Then, we consider the cases corresponding to other coordinates as discussed in cases 1-3. Finally,
we couple the remaining set arbitrarily.

\end{proof}

Using Lemma \ref{lemma:d>2}, one can easily prove a much simplified version of 
Lemmas \ref{lemma:alpha}, \ref{lemma:alpha2a}
and \ref{lemma:alpha2} with $L^{\gamma}, L^{\frac{1+\gamma}{2}},L^{\frac{3+\gamma}{4}}$
replaced by constants $K_1(\delta),K_2(\delta),K_3(\delta)$
This implies Proposition \ref{lemmaN,y} for the case $d>2$, $N=2$. 
Then the case of general $N$ follows the same
way as in $d=2$.

Finally, in case of 
Theorem \ref{thm2} we have assumed that $x_1 \neq 0$. Then choosing $\varepsilon < |x_1|$, 
the proof of the case of Theorem \ref{thm1} (with the choice $R$ is the constant $1$ function) applies.


\section{Proof of Theorem \ref{thm3} and Proposition \ref{prop:randomomega}}
\label{sec:thm3}

As in the proof of Theorems \ref{thm1} and \ref{thm2}, we prove the weak convergence
by showing that the moments converge. The latter one is showed by switching to the dual process.
Recalling some notation from Section \ref{sec:5},
we define $\bm{ Y}'$ from 
$\tilde{\bm Y} (t)$ the same way as we defined
$\bm{ Z}'$ from 
$\tilde{\bm Z} (t)$.
The proof consists of two parts. First we consider  
the case when the number of particles is $N=2$  and prove
\begin{proposition}
\label{lemma:d=1}
\[ \lim_{L \rightarrow \infty}\mathbb E \left[ T \left( \frac{{\bm Y}_1^{'(L)}(\infty)}{L} \right) \
T \left( \frac{{\bm Y}_2^{'(L)}(\infty)}{L} \right) 
\right] / \left[ u^{(L)}(x) \right]^2 =1
\]
\end{proposition}
Then we can derive an extension of Proposition \ref{lemma:d=1} to arbitrary $N$.
The idea of the proof of Proposition \ref{lemma:d=1} is borrowed from \cite{RY07} (the main difference
is in Lemma \ref{lemma:submartingale}). The extension to arbitrary $N$ is very similar to the 
argument in \cite{KMP82, LNY15}.

Let us write $P^{(L)}(A_1) = \mathbb P ({\bm Y}_1^{'(L)}  (\infty)= A_1)$ and
\[
P^{(L)}(A_1, A_2) = \mathbb P ({\bm Y}_1^{'(L)}  (\infty) = A_1, {\bm Y}_2^{'(L)}  (\infty)= A_2)
\]
for $A_1, A_2 \in \{ 0,L\}$.
The asymptotic hitting probabilities of one particle are given by

\begin{lemma}
\label{lemma:martingale}

 $\lim_{L \rightarrow \infty} \frac{P^{(L)}(L) }{\mathcal A^{(L)} (x)} = 1.$

\end{lemma}

\begin{proof}

Although this lemma follows from the connection between random walks and electrical networks,
we give a direct proof as its extensions will be needed later.
Note that by the definition of $\omega_0, \omega_L$ and $\psi(m)$, we have for any $1 \leq m \leq L-1$
\begin{equation}
\label{eq:mgale}
\frac{\omega_{m-1}}{ \omega_{m-1}+ \omega_m} r_{m-1/2} \psi(m) =
\frac{\omega_{m+1}}{ \omega_{m}+ \omega_{m+1}} r_{m+1/2} \psi(m+1).
\end{equation}
Let us write $\underline \Phi(m) = \sum_{i=1}^m \psi(i)$ for $0 \leq m \leq L$. 
Then (\ref{eq:mgale}) means that $\underline \Phi({\bm Y}'_1(k))$ is a bounded martingale.
After the first hitting of either $0$ or $L$, this martingale clearly stays constant. Then by the
martingale convergence theorem,
$$ P^{(L)}(L) \underline \Phi (L) = \underline \Phi ({\bm Y}'_1(0)).$$
Since $\mathcal A^{(L)}(x) \sim \underline \Phi ({\bm Y}'_1(0)) / \underline \Phi (L)$, Lemma 
\ref{lemma:martingale}  follows.
\end{proof}

Now with the notation 
 $\overline \Phi(m) = \sum_{i=m+1}^{L} \psi(i)$ for $0 \leq m \leq L$
our aim is to construct submartingales using  
$$ S_k := \underline \Phi({\bm Y}'_1(k)) \underline \Phi({\bm Y}'_2(k)) + 
\overline \Phi({\bm Y}'_1(k)) \overline \Phi({\bm Y}'_2(k)) 
$$
and
$$ T_k = \sum_{i=\min \{ {\bm Y}'_1(k), {\bm Y}'_2(k)\} +1}^{\max 
\{ {\bm Y}'_1(k), {\bm Y}'_2(k)\} } \psi(i).$$

\begin{lemma}
\label{lemma:subm}
There exists some constant $C$ only depending on the upper and lower bound of $r$ and $\omega$
such that $S_k + C T_k$ is a submartingale and $S_k - C T_k$ is a supermartingale.
\end{lemma}

\begin{proof}
Let us compute the conditional expectations with respect to $(\mathcal F_k)_k$,
the filtration generated by the process ${\bm Y}'$.
First, observe that $\mathbb E (S_{k+1} | \mathcal F_k) = S_k$ if 
$|{\bm Y}'_1(k)- {\bm Y}'_2(k)| \geq 2$. 
Next, by definition $\mathbb E (S_{k+1} | {\bm Y}'_1(k)={\bm Y}'_2(k)=i)$
is equal to 
\begin{eqnarray*}
&\frac{r_{i+1/2}}{r_{i-1/2}+r_{i+1/2}} 
&\int_0^1 p^2[\underline \Phi(i+1)^2 + \overline \Phi(i+1)^2] \\
&& + 2p(1-p)[\underline \Phi(i)\underline \Phi(i+1) + \overline \Phi(i)\overline \Phi(i+1)] \\
&&+(1-p)^2 [\underline \Phi(i)^2 + \overline \Phi(i)^2] d Beta(\omega_{i+1}/2, \omega_i/2,p)\\
&+ \frac{r_{i-1/2}}{r_{i-1/2}+r_{i+1/2}} 
&\int_0^1 p^2[\underline \Phi(i-1)^2 + \overline \Phi(i-1)^2] \\
&& + 2p(1-p)[\underline \Phi(i)\underline \Phi(i-1) + \overline \Phi(i)\overline \Phi(i-1)] \\
&&+(1-p)^2 [\underline \Phi(i)^2 + \overline \Phi(i)^2] d Beta(\omega_{i-1}/2, \omega_i/2,p)
\end{eqnarray*}
where we used the shorthand 
$d Beta(\alpha, \beta, p) = \frac{1}{B(\alpha, \beta)}p^{\alpha -1} (1-p)^{\beta -1} dp.$
Computing the integrals and using (\ref{eq:mgale}) gives
\begin{eqnarray*}
&& \mathbb E (S_{k+1} - S_k | {\bm Y}'_1(k)={\bm Y}'_2(k)=i)  \\
&=& 2 [\psi(i)]^2
\bigg( \frac{r_{i+1/2}}{r_{i-1/2}+r_{i+1/2}} \frac{\omega_{i+1}(\omega_{i+1}+2)}{(\omega_i+\omega_{i+1})(\omega_i+\omega_{i+1}+2)}\\
&&+ 
\frac{r_{i-1/2}}{r_{i-1/2}+r_{i+1/2}} \frac{\omega_{i-1}(\omega_{i-1}+2)}{(\omega_{i-1}+\omega_{i})(\omega_{i-1}+\omega_{i}+2)}
\bigg)
\end{eqnarray*}
Similarly, $\mathbb E (S_{k+1} - S_k| {\bm Y}'_1(k)=i, {\bm Y}'_2(k)=i+1) $
is by definition equal to 
\begin{eqnarray*}
&\frac{r_{i-1/2}}{r_{i-1/2}+r_{i+1/2}+r_{i+3/2}} &\frac{\omega_{i-1}}{\omega_{i-1} + \omega_i}
[-\psi(i)\underline \Phi(i+1) + \psi(i)\overline \Phi(i+1)] \\
&+\frac{r_{i+1/2}}{r_{i-1/2}+r_{i+1/2}+r_{i+3/2}}&
\big\{ \int_0^1 \big( p^2[\underline \Phi(i+1)^2 + \overline \Phi(i+1)^2] \\
&& + 2p(1-p)[\underline \Phi(i)\underline \Phi(i+1) + \overline \Phi(i+1)\overline \Phi(i)] \\
&&+(1-p)^2 [\underline \Phi(i)^2 + \overline \Phi(i)^2] - S_k \big) d Beta(\omega_{i-1}, \omega_i,p)\\
&& - \underline \Phi(i)^2 - \overline \Phi(i)^2 \big\}\\
&+\frac{r_{i+3/2}}{r_{i-1/2}+r_{i+1/2}+r_{i+3/2}} &\frac{\omega_{i+2}}{\omega_{i+1} + \omega_{i+2}}
[\underline \Phi(i)\psi (i+2) - \overline \Phi(i)\psi(i+2)].
\end{eqnarray*}
A similar computation to the previous one gives 
\begin{eqnarray*}
 &&\mathbb E (S_{k+1} - S_k | {\bm Y}'_1(k)=i,{\bm Y}'_2(k)=i+1) \\
 &=& - [\psi(i+1)]^2 \frac{2 r_{i+1/2}}{r_{i-1/2}+r_{i+1/2}+r_{i+3/2}} 
 \frac{\omega_{i}\omega_{i+1}}{(\omega_{i}+\omega_{i+1})(\omega_{i}+\omega_{i+1}+2)}.
\end{eqnarray*}
Just like in the case of $S_k$, we have
$\mathbb E (T_{k+1} | \mathcal F_k) = T_k$ if 
$|{\bm Y}'_1(k)- {\bm Y}'_2(k)| \geq 2$. Furthermore,
$$
\mathbb E (T_{k+1}  - T_k| {\bm Y}'_1(k)={\bm Y}'_2(k)=i) = 
[\psi(i+1)] \frac{ r_{i+1/2}}{r_{i-1/2}+r_{i+1/2}} 
 \frac{\omega_{i}\omega_{i+1}}{(\omega_{i}+\omega_{i+1})},
$$
and
$$
\mathbb E (T_{k+1} - T_k | {\bm Y}'_1(k)=i, {\bm Y}'_2(k)=i+1)  = 
[\psi(i+1)] \frac{1} {3}
 \frac{\omega_{i}\omega_{i+1}}{(\omega_{i}+\omega_{i+1} +2)}.
$$
We conclude that there is some positive constant $c$ such that 
\begin{itemize}
\item $0<\mathbb E (S_{k+1} - S_k | {\bm Y}'_1(k)  ={\bm Y}'_2(k)) < \frac{1}{c}$
\item $-1/c < \mathbb E (S_{k+1} - S_k | |{\bm Y}'_1(k) - {\bm Y}'_2(k)| = 1) <0$
\item $ c< \mathbb E (T_{k+1} - T_k | {\bm Y}'_1(k)  ={\bm Y}'_2(k))$
\item $c< \mathbb E (T_{k+1} - T_k | |{\bm Y}'_1(k) - {\bm Y}'_2(k)| = 1)$
\end{itemize}
The lemma follows.
\end{proof}

\begin{lemma}
\label{lemma:submartingale}
$\lim_{L \rightarrow \infty} \frac{P^{(L)}(L,L) + P^{(L)}(0,0)}{[\mathcal A^{(L)} (x)]^2 + 
[1-\mathcal A^{(L)} (x)]^2} = 1.$
\end{lemma}

\begin{proof}
Since $M_k:=S_k+CT_k$ is a bounded submartingale, we can apply the martingale convergence theorem
to deduce
$$  M_0 \leq \mathbb E(M_{\infty}) = 
(P^{(L)}(L,L) + P^{(L)}(0,0)) \underline \Phi^2 (L) + (P^{(L)}(0,L) + P^{(L)}(L,0))C\underline \Phi(L)$$
Since $\underline \Phi (L)/L$ is bounded away from zero and infinity, the lower bound follows. 
The upper bound is derived similarly from the fact that $S_k+CT_k$ is a supermartingale. 
\end{proof}

Now we are ready to prove Proposition \ref{lemma:d=1}.

\begin{proof}[Proof of Proposition \ref{lemma:d=1}]
By Lemma \ref{lemma:subsyst},
\[
P^{(L)}(L,L) = \frac12 [P^{(L)}(L) + P^{(L)}(L,L) + P^{(L)}(0,0) - P^{(L)}(0) ]
\]
Then by Lemmas \ref{lemma:martingale} and \ref{lemma:submartingale},
\[
P^{(L)}(L,L) \sim \left[  \mathcal A^{(L)}(x) \right]^2.
\]
Similarly, 
$P^{(L)}(0,0) \sim \left[  1- \mathcal A^{(L)}(x) \right]^2$,
$P^{(L)}(L,0) \sim P^{(L)}(0,L) \sim \left[  \mathcal A^{(L)}(x) \right][ 1- \mathcal A^{(L)}(x)]$.
Proposition \ref{lemma:d=1} follows.
\end{proof}

Since the extension of Proposition \ref{lemma:d=1} to arbitrary $N$ can be proved the same way as its analogues 
in
\cite{KMP82}, Section 3 and \cite{LNY15}, Section 6.2, we omit the proof here.
Thus we have finished the proof of Theorem \ref{thm3}.

Finally, Proposition \ref{prop:randomomega}(b) is elementary and 
Proposition \ref{prop:randomomega}(a)
follows from the connection between 
random walks and electrical networks (namely, from (\ref{eq:electric}) or 
Lemma \ref{lemma:martingale}) and the law of large numbers.

\section*{Appendix}

Here, we prove Lemma \ref{lemma:annuli}. We will use the notations of Section \ref{sec6}.
The proof uses similar coupling to the one in the proof of Lemma \ref{lemma:wasp}. 
This will suffice in case of (a), as $\| Z\|$ is small and in case of (d), as we only need
a weak estimate. However, we need to perform the coupling on a mesoscopic timescale to complete
the proof of (b) and (c).\\

{\bf Proof of (d)}.\\
We prove the following slightly stronger statement

\begin{centering}
{\it For every $\eta>0$ there is $\varepsilon = \varepsilon (\eta) >0$ such that
if $\log L - \log \log L < m < \log L + \log \varepsilon$ and $L$ is large enough, 
then $\mathcal E_{L,M} < \eta$.}
\end{centering}

With some fixed $C_1$, we couple 
$Z(k), k \in [s_0, (s_0 + C_1M^2)\wedge s_1]$ to a SSRW 
$W(k),  k \in [s_0, (s_0 + C_1M^2) \wedge s_1], W(s_0) = Z(s_0)$. We fix $C_1$, 
as we can by Lemma \ref{lemma:wasp}, such that
$\mathbb P(s_1 > s_0 + C_1M^2) < \eta/10$.
Second,  Lemma \ref{lemma:coupleRW} implies that by choosing $C_2 = C_2(C_1)$ large enough, 
\[ 
\mathbb P( \|W(k ) - Z( k) \| 
> C_2 M^3 / L^2  \text{ for some } k < (s_0 + C_1M^2)\wedge s_1
)  < \eta /10.
\]
Now let us define 
\begin{eqnarray*}
&&\underline s_1 = \min \{  k:  \| W(k) \| < (1/2 + \zeta) M \text{ or }  \| W(k) \| > (2 - \zeta) M \}, \\
&&\overline s_1 =  \min \{  k:  \| W(k) \| < (1/2 - \zeta) M \text{ or }  \| W(k) \| > (2 + \zeta) M \}
\end{eqnarray*}
Where $\zeta$ is fixed in such a way that the following events concerning a planar Brownian motion 
$\mathcal W$ have probability at least $1 - \eta /10$:  
\begin{enumerate}
\item[(A)] if $\| \mathcal W (0) \| = 1/2 + \zeta$ then $\mathcal W$ reaches $B(0, 1/2- \zeta)$ before reaching
$\mathbb R^2 \setminus B(0,2 - \zeta)$ and
\item[(B)] if $\| \mathcal W (0) \| = 2 - \zeta$ then $\mathcal W$ reaches 
$\mathbb R^2 \setminus B(0, 2 +  \zeta)$ before reaching
$B(0, 1/2 + \zeta)$.
\end{enumerate}
By the fact that $\log \| \mathcal W\|$ is martingale and by Donsker's theorem, we can also assume
$$| \mathbb P ( \| W (\underline s_1) \| < (1/2 + \zeta) M ) - 1/2 | < \eta/10, $$
possibly by further reducing $\zeta$ and by choosing $L$ large.
Finally, we choose $\varepsilon = \sqrt{ \zeta / C_2}$ so as $\zeta M >  C_2 M^3 /L^2$. Lemma \ref{lemma:annuli}(d) follows.\\

{\bf Proof of (a)}

If $2 \leq \| Z(k + s_0)\| $ for some $k < L^{13/10} \wedge \underline \tau_{3/5}$, 
then using Lemma \ref{lemma:coupleRW} we can couple 
$Z(k) - Z(k-1)$ to a step of a SSRW $W(k) - W(k-1)$ with probability $O(L^{-7/5})$.
Thus the probability 
$$ \mathbb P (Z(k) - Z(k-1) = W(k) - W(k-1) \text{ for all } k < L^{13/10}, 2 \leq \| Z(k + s_0)\|)$$
is $O(L^{-1/10})$. Now Lemma \ref{lemma:annuli}(a) follows from the proof of Lemmas
10 and 11 in \cite{LNY15}.\\

{\bf Proof of (b) and (c)}

The idea of the proof of these cases is borrowed from \cite{DSzV08}.
First, we fix some positive $\xi <3/20$, write $K= \lfloor L^{\xi} \rfloor $
and consider the process $(Z(s_0+ jK))_{j \geq 1}$,
stopped upon reaching $B(0, M/2 - K)$ or $\mathbb R^2 \setminus B(0, 2M + K)$.
Let us denote by $\bar s$ the corresponding stopping time, i.e. $\bar s = sK$ with the smallest $s$
such that $Z(s_0+ sK)$ is stopped.  We will also write
 $\mathfrak z = Z(s_0+ jK) - Z(s_0+ (j-1)K)$ and 
$\mathfrak x_j = \log \| Z(s_0+ jK) \|^2$. Now we have
\[
\| Z(s_0+ (j+1)K) \|^2 = \| Z(s_0+ jK) \|^2 + 2(Z(s_0+ jK), \mathfrak z)
+ \| \mathfrak z,\mathfrak z \|^2 
\]
and
\[
\mathfrak x_{j+1} - \mathfrak x_j =\log \left( 1+ 
\frac{2(Z(s_0+ jK), \mathfrak z) +\| \mathfrak z,\mathfrak z \|^2 }{\| Z(s_0+ jK) \|^2}  \right)
\]
By Taylor expansion,
\begin{eqnarray*}
&& \mathfrak x_{j+1} - \mathfrak x_j \\
&=& \frac{2(Z(s_0+ jK), \mathfrak z)}{\| Z(s_0+ jK) \|^2}
+
\frac{\| \mathfrak z\|^2}{\| Z(s_0+ jK) \|^2}
-\frac{2(Z(s_0+ jK), \mathfrak z)^2}{\| Z(s_0+ jK) \|^4} + O(L^{3 \xi} M^{-3})\\
&=& I + II + III +  O(L^{3 \xi} M^{-3})
\end{eqnarray*}
As before, we couple $\mathfrak z$ to a SSRW $W(K)$ by Lemma 
\ref{lemma:coupleRW} such that
if $\mathfrak y = \mathfrak z - W(K) $, then for all positive integer $n$,
\[ 
\mathbb P ( \| \mathfrak y \| =k ) = O(L^{n\xi -2n} M^n)
 \]
If $\mathcal F_k $ is the filtration generated by $\tilde{\bm Z}_i (s_0 + l), l \leq k, i=1,2$, then
\[
 \mathbb E \left( I  \right) =
 \mathbb E \left(
 \frac{ \mathbb E \left(  {2(Z(s_0+ jK), W(K))}  |
\mathcal F_j \right)}{\| Z(s_0+ jK) \|^2}  \right)
+
\mathbb E \left(
\frac{ \mathbb E \left(  {2(Z(s_0+ jK), \mathfrak y)} |
\mathcal F_j \right)}{\| Z(s_0+ jK) \|^2}  \right)
\]
Here, the first term is zero by symmetry of the SSRW $W$, and the second is estimated by the 
Cauchy-Schwarz inequality in cases $1 \leq \| \mathfrak y \| \leq 2$ and $\| \mathfrak y \| \geq 3 $.
We conclude $\mathbb E (I) = L^{\xi -2}.$
A similar argument shows that
\begin{equation}
\label{eq:II}
\mathbb E ( \| \mathfrak z\|^2 ) = K + O(M L^{3\xi -2}).
\end{equation}
The above leading term is 
$\mathbb E ( \| W (K) \|^2) = K$, which can be easily proved by leveraging the fact that
$(W_{(1)}(k) - W_{(1)}(k-1))_{1 \leq k \leq K}$ and $(W_{(2)}(k) - W_{(2)}(k-1))_{1 \leq k \leq K}$  
are two independent one dimensional SSRW's, where 
$v_{(1)} = ((1/\sqrt 2, 1/\sqrt 2), v)$ and $v_{(2)} = ((1/\sqrt 2, - 1/\sqrt 2), v)$.
Finally,
\begin{eqnarray*}
&&\mathbb E (
2 (Z(s_0+ jK), \mathfrak z)^2 | \mathcal F_j) \\
&=& \sum_{i,j=1}^2
2  Z_{(i)}(s_0+ jK)  Z_{(j)}(s_0+ jK) \mathbb E (\mathfrak z_{(i)} \mathfrak z_{(j)} ) \\
&=&
\sum_{i=1}^2
2  (Z_{(i)}(s_0+ jK))^2 \mathbb E [(W_{(i)}(K))^2] + O(M^2 L^{3 \xi -2}) \\
&=& \| Z(s_0+ jK) \|^2K +  O(M^2 L^{3 \xi -2}).
 \end{eqnarray*}
We conclude that 
\begin{equation}
\label{eq:p23last}
\mathbb E ((\mathfrak x_{j+1} - \mathfrak x_j ) 1_{\bar s > jK}) = O(L^{3 \xi} M^{-3} + L^{\xi -2}).
\end{equation}
Thus
\begin{equation}
\label{eq:p24_1}
\mathbb E (\mathfrak x_{\min \{ \bar s/ K, A/ K \}}  ) = 2m + O\left( A L^{2 \xi} M^{-3} +A_M L^{-2}\right)
\end{equation}
Recall that we are given some $\gamma <1$.  
\begin{itemize}
\item In case of Lemma \ref{lemma:annuli}(b), let us choose $A = M^{2+ \frac{1-\gamma}{\gamma}}$.
Assuming, as we can that $\xi < 1-\gamma$ gives $A L^{-2}< L^{-\xi}$. If $\xi$ is small (say smaller than
$1/20$), then $ A L^{2 \xi} M^{-3}< L^{-\xi}$ holds as well.
\item In case of Lemma \ref{lemma:annuli}(c) let us choose $A =L^2 \log^{-3/2} L$. Since both
$1-\gamma$ and $\xi$ are small, we have $A L^{2 \xi} M^{-3} +A_M L^{-2} = O(\log^{-3/2} L)$.
\end{itemize}

Thus the error
term in (\ref{eq:p24_1}) can be replaced by $O(E)$, 
where $E =L^{-\xi}$ in case of Lemma \ref{lemma:annuli}(b) and $E = \log^{-3/2} L$ in case of
Lemma \ref{lemma:annuli}(c).
Next, by Lemma \ref{lemma:wasp}, 
\begin{equation}
\label{eq:p24_2}
\mathfrak x_{\min \{ \bar s/ K, A/ K \}} = 2 m \pm 2 + O(L^{\xi} M^{-1}) = 
2 m \pm 2 + O(L^{-\xi})
\end{equation}
holds except on a set of probability $O(\theta^{A/M^2})$. 
On this exceptional set, we bound the left hand side of (\ref{eq:p24_2}) by $O(\log M)$.
In case of Lemma \ref{lemma:annuli}(b),
$ \theta^{A/M^2} \log M$ is superpolynomially small in $L$. In case of Lemma \ref{lemma:annuli}(c),
$ \theta^{A/M^2} \log M< \theta^{\sqrt{\log L}} \log L= O(\log^{-3/2} L)$. 
It follows that
\[ 
\mathbb P (\| Z_{s_0 + \bar s K} \| \leq M/2 - K ) = \frac12 - O(E)
\]
Since $\| Z_{s_0 + \bar s K} \| \leq M/2 - K$ implies $Z(s_1) \in \mathcal C_{m-1}$, we conclude
\[ 
\mathbb P ( Z(s_1) \in \mathcal C_{m-1} ) \geq \frac12 - O(E).
\]
An analogous argument shows
\[ 
\mathbb P ( Z(s_1) \in \mathcal C_{m+1} ) \geq \frac12 - O(E).
\]
Lemma \ref{lemma:annuli}(b) and (c) follow.

\section*{Acknowledgements}
I am highly indebted to Imre P\'eter T\'oth for raising a question which initiated this research.
I am also grateful to Lai-Sang Young for discussions and encouragement.
This project began while I was affiliated with the Courant Institute, New York University.

\end{document}